\documentclass[a4paper,12pt]{article}

\usepackage{amsmath,pgfplots,caption}
\usepackage{amssymb}
\usepackage{amsfonts}
\usepackage{graphicx,epstopdf}
\usepackage{mathrsfs}
\usepackage{color}
\usepackage{subfigure}

\usepackage{ulem} 

\newtheorem{theorem}{Theorem}
\newtheorem{notations}[theorem]{Notations}   
\newtheorem{proposition}[theorem]{Proposition}   
\newtheorem{lemma}[theorem]{Lemma}

\newtheorem{example}[theorem]{Example}

\newenvironment{proof}[1][]{\noindent {\bf Proof #1:\;}}{\hfill $\Box$}

\newcommand{\bigslant}[2]{{\raisebox{.2em}{$#1$}\left/\raisebox{-.2em}{$#2$}\right.}}

\def\QQ{{\mathbb{Q}}} 
\def\RR{{\mathbb{R}}} 
\def\CC{{\mathbb{C}}} 

\def\f{{f}}
\def\g{{g}}

\def\q{{q}}

\def\fermU{{B}} 
\def\fermP{{\mathcal{P}}} 
\def\fermQ{{\mathcal{Q}}} 

\def\scC{{\mathscr{C}}}
\def\scS{{\mathscr{S}}}
\def\sfP{{\mathsf{P}}}
\def\sfQ{{\mathsf{Q}}}

\def\sing{{\rm sing}\:}
\def\crit{{\rm crit}\:}
\def\reg{{\rm reg}\:}

\newcommand{\zeroset}[1]{{Z(#1)}}
\newcommand{\ideal}[1]{{I(#1)}}

\def\setD{{\mathcal{D}}}
\def\setC{{\mathcal{C}}}
\def\setV{{\mathcal{V}}}
\def\setW{{\mathcal{W}}}
\def\setZ{{\mathcal{Z}}}

\def\zarA{{\mathscr{A}}}
\def\zarfiber{{\mathscr{W}}}
\def\zarU{{\mathscr{U}}}
\def\zarV{{\mathscr{V}}}
\def\zarM{{\mathscr{M}}}
\def\zarO{{\mathscr{O}}}

\def\sfG{{\mathsf{G}}} 
\def\sfH{{\mathsf{H}}} 

\def\A{{A}} 
\def\M{{M}} 
\def\jac{{D}} 

\def\rank{{\rm rank}}
\def\deg{{\rm deg}}

\def\GL{{\mathrm{GL}}} 

\def\X{{x}} 
\def\Y{{y}} 
\def\Z{{z}} 

\def\x{{x}} 
\def\y{{y}} 
\def\z{{z}} 
\def\v{{v}}  
\def\fiber{{w}}  

\def\fraka{{a}}

\def\frakA{{\A}}

\def\softO{\ensuremath{{O}{\,\tilde{ }\,}}}

\title{Real root finding for determinants of linear matrices} 

\begin{document}

\author{Didier Henrion$^{1,2,3}$ \and Simone Naldi$^{1,2}$ \and Mohab Safey El Din$^{4,5,6,7}$}

\footnotetext[1]{CNRS, LAAS, 7 avenue du colonel Roche, F-31400 Toulouse; France.}
\footnotetext[2]{Universit\'e de Toulouse; LAAS, F-31400 Toulouse, France.}
\footnotetext[3]{Faculty of Electrical Engineering, Czech Technical University in Prague, Czech Republic.}
\footnotetext[4]{Sorbonne Universit\'es, UPMC Univ Paris 06, Equipe PolSys, LIP6, F-75005, Paris, France.}
\footnotetext[5]{INRIA Paris-Rocquencourt, PolSys Project, France.}
\footnotetext[6]{CNRS, UMR 7606, LIP6, France.}
\footnotetext[7]{Institut Universitaire de France.}

\date{\today}

\maketitle

\begin{abstract}
Let $\A_0, \A_1, \ldots, \A_n$ be given square matrices of size $m$ with
rational coefficients.
The paper focuses on the exact computation of one point in each
connected component of the real determinantal variety
$\{\X \in\RR^n \: :\: \det(\A_0+x_1\A_1+\cdots+x_n\A_n)=0\}$.
Such a problem finds applications in many areas such as control theory, computational
geometry, optimization, etc. Using standard complexity results
this problem can be solved using
$m^{O(n)}$ arithmetic operations.
Under some genericity assumptions on the coefficients of the matrices,
we provide an algorithm solving this problem whose runtime
is essentially quadratic in
${{n+m}\choose{n}}^{3}$. We also report on experiments with a computer
implementation of this algorithm. Its practical performance illustrates the complexity
estimates. In particular, we emphasize that for subfamilies
of this problem where $m$ is fixed, the complexity is polynomial in $n$.
\end{abstract}

\begin{center}
\small{\bf Keywords}\\[.5em]
Computer algebra, real algebraic geometry, determinantal varieties.
\end{center}

\section{Introduction}\label{sec1}

\subsection{Problem statement}

Let $\A_0, \A_1, \ldots, \A_n$ be given square matrices of size $m$ with
coefficients in the field of rationals $\QQ$.
Consider the  affine map  defined as 
$$
\X=(\X_1,\ldots, \X_n) \mapsto \A(\X) = \A_0 + \X_1 \A_1 + \cdots + \X_n \A_n.
$$
Consistently with the technical literature, we use the terminology
\textit{linear matrix} to refer to $\A(\X)$, even though the constant
term $\A_0$ is not necessarily zero.
The determinant of $\A(\X)$, denoted by $\det\A(\X)$, lies in the polynomial ring
$\QQ[\X]$ and it has degree at most $m$. This polynomial
defines the complex \textit{determinantal variety}
\[
{\setD}=\big\{\X\in\CC^n \: :\: \det \A(\X)=0\big\}.
\]
In other words, $\setD \subset \CC^n$ is the set of complex vectors $\X$ at which $\rank\A(\X)\leq m-1$.
The goal of this paper is to provide a  \textit{computer algebra algorithm} with explicit complexity estimates for
computing at least one point in each connected component of the
\textit{real} determinantal variety $\setD\cap \RR^n$.

\subsection{Motivations}

First notice that when $n=1$ our problem is called the real algebraic eigenvalue problem
\cite{Wilkinson65}, and hence that the case $n>1$ can be seen as a multivariate generalization.

Non-symmetric square matrices depending  linearly on parameters arise in
many problems of systems control and signal processing. For example,
the Hurwitz matrix is used in stability criteria for systems described
by linear ordinary differential equations, and vanishing of the
determinant of the Hurwitz matrix corresponds to a bifurcation
between stability and instability, see e.g. \cite{Barmish94}.
Alternatively, finding points on the real determinantal variety of the
Hurwitz matrix amounts to finding parameters (e.g. corresponding
to a feedback control law, or to structured uncertainty affecting the system)
corresponding to a system configuration at the border of stability.

Another classical example of non-symmetric square linear matrix arising
in signals and systems is the Sylvester matrix ruling controllability of
a linear differential equation. In this context, vanishing of the determinant of the
Sylvester matrix corresponds to a loss of controllability of the
underlying system \cite{Kucera79}.

Linear matrices and optimization on determinantal varieties arise
also in statistics \cite{draismarodriguez2013,hauenstein2012maximum}
and in computational algebraic geometry \cite{euclideandegree2013}.

Under the assumption that the matrices $\A_0, \ldots, \A_n$
are symmetric, the matrix $\A(\X)$ is symmetric, and hence it has
only real eigenvalues for all $\X \in \RR^n$. The condition
$\A(\X)\succeq 0$, meaning that $\A(\X)$ is positive semidefinite,
is called a \textit{linear matrix inequality}, or LMI. It is a convex condition
on the space of variables $\X$ which appears frequently in diverse problems of
applied mathematics and especially in systems control theory, see e.g. \cite{befb94}. 
A classical example is the Lyapunov stability condition for a linear
ordinary differential equation which is an LMI in the parameters
of a Lyapunov function (a certificate or proof of stability)
depending quadratically on the system state.

\begin{figure}[!ht]
\centering
\includegraphics[width=0.7\textwidth]{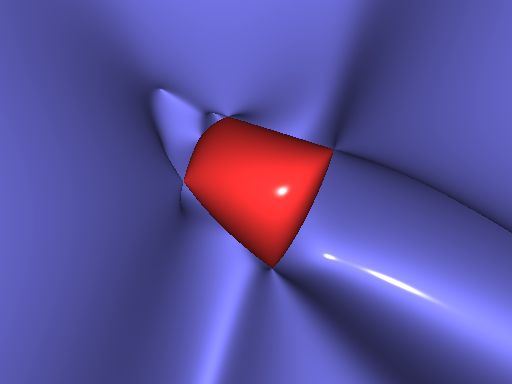}
\caption{A spectrahedron (red) with its real determinantal
variety.\label{fig:spectrahedron}}
\end{figure}

When the matrix $\A(\X)$ is symmetric, the set
 \begin{equation}\label{spectra}
\mathcal{S}=\big\{x\in\RR^n \ \big| \ A(x)\succeq 0\big\}
\end{equation}
is called a \textit{spectrahedron}. Spectrahedra are affine sections
of the cone of positive semidefinite matrices and they represent
closed convex basic semialgebraic sets, i.e. convex sets
that can be defined by the common nonnegativity locus of a finite
set of polynomials; they are the object of active studies
mainly in optimization theory, real algebraic geometry and
control theory \cite{l09,l10,bpt13}. Following a question posed
in \cite[Section 4.3.1]{n06}, the authors of \cite{hn09} conjectured
that every convex semialgebraic set is the projection of
a spectrahedron. On Figure \ref{fig:spectrahedron} is represented
a spectrahedron (for $n=3$ and $m=5$) together with its real determinantal variety.

The minimization of a given function, for example a polynomial, over
real convex sets is a central problem in optimization theory. If the
function is linear and the set is a spectrahedron, this is exactly
the aim of \textit{semidefinite programming} (SDP), see \cite{bn01}.
If the input data of a semidefinite program are defined over $\QQ$,
the solutions are algebraic numbers, and the authors in
\cite{sturmfelsnieranestad} investigated their algebraic degree:
giving explicit formulas or bounds for this value is a measure
of the complexity of the given program.

Convex LMIs and SDP are also widely used for solving
nonconvex polynomial optimization problems. Indeed, these optimization problems
are linearized in the space of moments of nonnegative measures (which is infinite-dimensional)
and a suitable sequence of LMI relaxations, or truncations (the so-called \textit{Lasserre hierarchy}),
that can be solved via SDP, provides the solution to the original
problem. The feasible set of every truncated problem is a
spectrahedron in the space of moments, for more details
see \cite{l09,l10} and references therein.

So far, the problem of (deciding the existence and) computing such
solutions has been addressed via several numerical methods, the most
successful of which are primal-dual interior-point algorithms \cite{bn01}
implemented in \textit{floating-point arithmetic} in different SDP solvers
\cite{m12}.

In this paper, the problem of computing points on spectrahedra is
linked to polynomial systems solving over the reals. In fact, we
are interested in the developement of an \textit{exact computer algebra}
algorithm to compute
real points on hypersurfaces defined by the zero locus of
determinants of affine matrix expressions. By exact algorithm we mean
that we do not content ourselves with approximate floating-point
computations.
Our main motivation starts from the geometrical aspects of SDP as explained above:
boundaries of spectrahedra are subsets of determinantal hypersurfaces,
and so solving this problem \textit{efficiently} is a necessary step to address
the associated positivity problem $\A(\X)\succeq 0$, since the rank of the matrix
$\A(\X)$ at a point in the boundary of the spectrahedron $\mathcal{S}$
drops at least by one, while a point in the interior corresponds to a positive
definite matrix. Finding such a point is a certificate of strict feasibility.

\subsection{State of the art}

Modern computer algebra algorithms for solving our problem require at most $m^{O(n)}$
arithmetic operations in $\QQ$, see \cite[Ch.11, Par.6]{BaPoRo06} and
references therein. The core idea is to reduce the input problem to a
polynomial optimization problem whose set of optimizers is
expected to be finite and to meet every connected component of the
solution set under study. Such a reduction must be done carefully,
especially for unbounded sets or singular situations. So far, it is an
open problem to get a competitive implementation of the algorithms in
\cite{BaPoRo06}: unbounded and singular cases imply algebraic
manipulations that have no impact on the complexity class but require
to work over Puiseux series fields, and this increases the constant hidden
by the big-$O$ notation in the exponent.

During the past decade, tremendous efforts have been made to obtain
algorithms that are essentially quadratic in $m^n$ when dealing with
one $n$-variate polynomial equation of degree $m$, see e.g. \cite{BGHM1, BGHM2,
  BGHM3, BGHM4, SaSc03}. The goal is to get an implementation that
reflects the theoretical complexity gains. Most of these algorithms are
probabilistic: some random choices independent
of the input are performed to ensure genericity properties.
Our contribution shares these features and it is inspired by some
geometric ideas in \cite{SaSc03}.

A main limitation is that the algorithms in \cite{SaSc03} are dedicated to the
smooth case.  In our case, it turns out that $\setD$ is in general a
singular variety -- see e.g. \cite{bruns1988determinantal}
and recall Figure \ref{fig:spectrahedron} -- which
makes our problem more difficult from a geometric point of view.

Algorithms in \cite{BGHP10} deal with singular situations but do not
return sample points in the connected components that are contained in
the singular locus of the variety. As a consequence, one
cannot use them to decide the emptiness of $\setD \cap \RR^n$. The
algorithm in \cite{S05} may be used but it suffers from an extra-cost,
since it requires essentially $m^{4n}$ arithmetic operations.

Moreover, in \cite{FauSafSpa10a, faugere2012critical, FauSafSpa13},
the authors have developed algorithms and complexity estimates to
isolate the real solutions of determinantal systems (see also
\cite{FauSafSpa11} for related works on a bilinear setting). Beyond
the interest of solving our problem for the aforementioned
applications, it is of interest to extend these works to the real and
positive dimensional case.

In practice, one can observe that, when a determinantal equation is
given as input to software implementing singly exponential algorithms
\cite{raglib}, its behaviour is significantly different and worse than
the one observed on generic equations.

\subsection{Basic definitions}\label{subsection:basicdefn}

Before describing the main results of this paper and the basic ideas
on which they rely, we need to introduce some notations and basic
definitions that are used further. We refer to
\cite{perrin2008algebraic,Shafarevich77} for details.

We use the notations $\QQ_* = \QQ\setminus\{0\}$  and $\CC_* = \CC\setminus\{0\}$.
We also denote $\CC^m$ the set of complex vectors of length $m$ and $\CC_*^m = \CC^m
\setminus \{0\}$. Given two vectors $x,y \in \CC^m$, with $x'y$ we denote their
scalar product $x_1y_1+\ldots+x_my_m$.

A subset $\setV \subset \CC^n$ is said to be an {\it affine algebraic variety}
defined over $\QQ$ if there exists a system (i.e. a finite set) of polynomials
$\f=(f_1, \ldots, f_p) \in \QQ[\X]^p$ such that $\setV$
is the locus of their common complex solutions, i.e.
$\setV = \{\X \in \CC^n \: :\: f(\X) = 0\} = \{\X \in \CC^n \: :\: f_1(\X) = \cdots = f_p(\X) = 0\}$.
In this case we write $\setV=\zeroset{\f}=\f^{-1}(0)$. Algebraic varieties are the closed sets in
the Zariski topology, hence any set defined by a polynomial
inequation $f\neq 0$ defines an open set for the Zariski topology.
We also consider the closure $\overline{\setV}$
of a set $\setV \subset \CC^n$ for the Zariski topology, that is the
smallest algebraic subset of $\CC^n$ containing $\setV$.

The set of polynomials that vanish on an
algebraic set $\setV$ generates an ideal of $\QQ[\X]$
associated to $\setV$, denoted by $\ideal{\setV}$.
This ideal is radical (i.e. $f^k\in \ideal{\setV}$ for some integer $k$
implies that $f\in \ideal{\setV}$) and it is generated by a finite set
of polynomials, say $\f=(f_1, \ldots, f_p)$, and we write
$\ideal{\setV} = \langle f_1,\ldots,f_p \rangle = \langle f \rangle$.

Let $\setV \subset \CC^n$ be an affine algebraic variety. Then the quotient
ring $\CC[\setV]=\CC[\X]/I(\setV)$ is the {\it coordinate ring}
of the variety $\setV$: the elements of $\CC[\setV]$ are called {\it regular
functions} on $\setV$. A map $f \colon \setV \to \setW \subset \CC^p$
defined over $\setV$ with values in $\setW$, such that $f \in \CC[\setV]^p$,
is called a {\it regular map}, and if $f$ is a bijection and its inverse
is also a regular map, then $f$ is an {\it isomorphism} of affine
algebraic varieties.

Let $\GL(n,\QQ)$ denote the set of non-singular matrices of size $n$
with coefficients in $\QQ$. Its identity matrix is denoted by $Id_n$.
Given a matrix $\M \in \GL(n,\QQ)$ and a polynomial system
$\X \in \CC^n \mapsto \f(\X) \in \CC^p$
we denote by $\f\circ\M$ the polynomial system
$\X \in \CC^n \mapsto \f(\M\X) \in \CC^p$.
If $\setV = \zeroset{\f}$, the image set
$\zeroset{\f\circ\M}=\{\X \in \CC^n : \f(\M\X) = 0\}
=\{\M^{-1}x \in \CC^n : \f(\X) = 0\}$ is denoted by
$M^{-1}\setV$.

Let
\[
\left(\frac{\partial f}{\partial \X_k}\right) = 
\left(\begin{array}{c}
\frac{\partial f_1}{\partial \X_k} \\
\vdots \\
\frac{\partial f_p}{\partial \X_k}
\end{array}\right)
\]
denote the vector of $\QQ[\X]^p$ containing partial derivatives of $f$ w.r.t. variable $\X_k$,
for some $k=1,\ldots,n$.
The co-dimension $c$ of $\setV$ is the
maximum rank of the Jacobian matrix 
$$
\jac\f = \left(\frac{\partial f}{\partial \X_k}\right)_{k=1,\ldots,n} =
\left(
\begin{array}[]{cccc}
  \frac{\partial f_1}{\partial \X_1} & \ldots & \ldots&  \frac{\partial f_1}{\partial \X_n}\\
\vdots & & & \vdots \\
  \frac{\partial f_p}{\partial \X_1} & \ldots & \ldots&  \frac{\partial f_p}{\partial \X_n}\\
\end{array}
\right)
$$
evaluated at $\X \in \setV$. The dimension of $\setV \subset \CC^n$ is
$n-c$. 

Let $\setV\subset \CC^n$ be an algebraic set. We say that $\setV$ is
{\it irreducible} if it is not the union of two sets that are closed for the
Zariski topology and strictly contained in $\setV$.  Otherwise
$\setV$ is the union of finitely many irreducible algebraic
sets, its {\it irreducible components}.

Most of the time, we will consider {\it equidimensional} algebraic
sets: these are algebraic sets whose irreducible components share
the same dimension. An algebraic set $\setV$ of dimension $d$ is the
union of equidimensional sets of dimensions $k=0,1,\ldots,d$:
this is the so-called equidimensional decomposition of $\setV$.
Suppose that $\setV$ is $d$-equidimensional, that is, equidimensional of dimension $d$.
Given a polynomial system $\f : \CC^n \to \CC^p$ and a point $\X \in \setV=\zeroset{\f}$, we say that $\X$ is {\it regular} if
$\jac\f(\X)$ has rank $n-d$, and {\it singular} otherwise. An algebraic set 
whose points are all regular is called {\it smooth}, and {\it singular} otherwise.
The set of singular points of an algebraic set $\setV$ is denoted
by $\sing\setV$, while the set of its regular points is denoted
by $\reg\setV$.

Given  a polynomial system $\f  : \CC^n \to \CC^p$, suppose that $\setV=\zeroset{\f} \subset\CC^n$
is a smooth $d$-equidimensional algebraic set, and let $g : \CC^n \to \CC^m$ be a polynomial system.
Then the set of {\it critical points} of the restriction of $g$ to $\setV$ is defined by the zero
set of $\f$ and the minors of size $n-d+m$ of the matrix
\[
\left(\begin{array}{c}
\jac\f \\
\jac g
\end{array}\right)
\]
and we denote it by $\text{crit}(g,f)$. In particular, the critical points of the restriction
to $\setV$ of the \textit{projection map} $\pi_i \colon (x_1,\ldots,x_n) \mapsto (x_1,\ldots,x_i)$
is the zero set of $\f$ and the minors of size $n-d$ of the truncated Jacobian
\[
\left(\frac{\partial f}{\partial \X_k}\right)_{k=i+1,\ldots,n} 
\]
obtained by removing the first $i$ columns in the Jacobian of $f$.
The same definition applies to the equidimensional components of a generic algebraic set.

\subsection{Data representation and genericity assumptions}

\subsubsection{Input}

We assume that the linear matrix $\A(\X) = \A_0+\X_1\A_1+\cdots+\X_n\A_n$ is described
via the square matrices $\A_0, \A_1,\ldots,\A_n$ of size $m$ with coefficients in $\QQ$,
which can also be understood as a point in $\QQ^{(n+1)m^2}$.
To refer to this input we use the short-hand notation $\A$.

\subsubsection{Output}\label{ratpar}

Our goal is to compute exactly
sample points in each connected component of the real variety $\setD\cap\RR^n$.
Our algorithm  consists of reducing the initial problem to isolating
the real solutions of an algebraic set $\setZ\subset \CC^n$ of dimension
at most $0$. To this end, we compute a {\it rational parametrization} of
$\setZ$ that is given by a polynomial  system $\q=(q_0, q_1, \ldots, q_n, q_{n+1}) \in \QQ[t]^{n+2}$ such that
$q_0$, $q_{n+1}$ are coprime (i.e. with constant greatest common divisor) and
\[
\setZ = \left\{ x = \left(\frac{q_1(t)}{q_0(t)},\cdots,\frac{q_n(t)}{q_0(t)}\right) \in \CC^n \: :\: q_{n+1}(t)=0 \right\}.
\]
This allows to reduce real root counting isolation to a univariate
problem. Note also that the cardinality of $\setZ$ is the degree of polynomial
$q_{n+1}$, provided it is square-free; we denote it by $\deg\:\q$.
 
Given a polynomial system defining a finite algebraic set
$\setZ\subset \CC^n$, there exist many algorithms for computing such a
parametrization of $\setZ$. In the experiments reported in Section
\ref{sec:experiments}, we use implementations of algorithms based on Gr\"obner
bases \cite{F4, F5} and the so-called change of ordering algorithms
\cite{fglm, newfglm} because they have the current best practical
behavior. Nevertheless, our complexity analyses are based on the
geometric resolution algorithm given in \cite{GiLeSa01}.

\subsection{Genericity assumptions}

\subsubsection{Singular locus of the determinant} \label{remark:sing}

Throughout the paper, we suppose that the singular locus of the
algebraic set $\setD = \{x \in \CC^n : \det A(x)=0\}$ is included in
the set $\{x \in \CC^n : \rank \, A(x) \leq m-2\}$, that is if $x \in
\setD$ is such that $\rank \, A(x) = m-1$, then $D_x \det A \neq 0$ at
$x$. This property is generic in the space of input matrices
$\CC^{m^2(n+1)}$. This fact can be proved easily via Bertini's theorem
(see for example the proof of \cite[Theorem 5.50]{bpt13}).

In the sequel, we say that $\A$ satisfies $\sfG_1$ if the singular
locus of $\det A$ is exactly the set of points at which the rank
deficiency of this matrix is greater than $1$.

\subsubsection{Smoothness and equidimensionality}\label{sec:genericityinput}

We define
\[
\begin{array}{lrcl}
\f(\A) : & \CC^{n+m} & \to & \CC^m \\
& (\X,\Y) & \mapsto & \A(\X)\Y
\end{array}
\]
as a polynomial system of size $m$ in the variables $\X=(\X_1,\ldots,\X_n)$ and $\Y=(\Y_1, \ldots, \Y_m)$.
Given ${u}=(u_1, \ldots, u_m, u_{m+1})  \in \CC^{m+1}$ with $u_{m+1} \neq 0$, define
\[
\begin{array}{lrcl}
\f(\A,u) : & \CC^{n+m} & \to & \CC^{m+1} \\
& (\X,\Y) & \mapsto &  (\A(\X)\Y,\: u'(y,-1))
\end{array}
\]
where $u'(y,-1) = u_1y_1 + \cdots + u_my_m-u_{m+1}$ denotes the inner
product of vectors $u$ and $(y,-1) \in \CC^{m+1}$, and let $\setV(\A,
u)= \zeroset{\f(\A,u)} \subset \CC^{n+m}$. 

In the sequel, we will
assume that the system $\f(\A,u)$ satisfies the following assumption
$\sfG_2$.

{\bf Assumption $\sfG_2$}.
We say that a polynomial system $f \in \QQ[x]^p$ satisfies $\sfG_2$ if
\begin{itemize}
\item $\langle \f \rangle$ is radical, and
\item $\zeroset{\f}$ is either empty or smooth and equidimensional of
  co-dimension $p$. \smallskip
\end{itemize}

We also say that a linear map $\A$ satisfies $\sfG_2$ if the polynomial
system $\f(\A,u)$ satisfies $\sfG_2$ for all $u \in \CC^m_*$. We will
prove later that for a generic choice of $\A$, this property holds.

We say that $\A$ satisfies $\sfG$ when it satisfies $\sfG_1$ and $\sfG_2$. 

\subsection{Main results and organization of the paper}

The main result of the paper is sketched in the following. Its detailed
statement is in Proposition \ref{prop:complexity:RealDet} and it will be proved
in Section \ref{ssec:algo:complexity}.

{\it Let $\A_0, \A_1, \ldots, \A_n$ be square matrices of size $m$
  with coefficients in $\QQ$ satsifying the above genericity
  assumptions.  There exists a probabilistic exact algorithm with
  input $\A_0, \A_1, \ldots, \A_n$ and output a rational
  parametrization encoding a finite set of points with non-empty
  intersection with each connected component of $\setD\cap\RR^n$.

In case of success, the complexity of the algorithm is within
$$
\softO
\left(
n^2m^2(n+m)^5\binom{n+m}{n}^6
\right)
$$
arithmetic operations, where $\softO(s)={O}(s \log^k s)$ for some
$k \in \mathbb{N}$.
}

We also analyze the practical behaviour of a first implementation of
this algorithm. Our experiments show that it outperforms the
state-of-the-art implementations of general algorithms for grabbing
sample points in real algebraic sets.

The paper is organized as follows.

{\it Section} \ref{sec:algo} contains a detailed description of the algorithm
and of its subroutines. Moreover, its formal description is provided. Section
\ref{ssec:algo:correctness} contains all regularity
results, that is Propositions \ref{prop:smoothness}, \ref{prop:dimlag} and
\ref{prop:closedness}, proved in the following sections. It also contains
the proof of correctness of the algorithm (Theorem \ref{theo:correctness}). As already mentioned,
the proof of the main result is given in Section \ref{ssec:algo:complexity}.
{\it Section} \ref{sec:smoothness} contains the proof of Proposition 
\ref{prop:smoothness}. {\it Section} \ref{sec:dimension} contains the proof
of Proposition \ref{prop:dimlag}. {\it Section} \ref{sec:closedness} contains
the proof of Proposition \ref{prop:closedness}. Finally, {\it Section} \ref{sec:experiments}
contains numerical data of practical experiments and some examples.

\section{Algorithm: correctness and complexity}\label{sec:algo}

\subsection{Description of the algorithm}\label{ssec:algo:description}

Our algorithm is guaranteed to return an output under some genericity
assumptions on the input. If the genericity assumptions are not satisfied,
the algorithm raises an error.  The algorithm consists of computing critical points
of the restriction of linear projections to a given algebraic variety after a randomly chosen
linear change of variables. These points are  the solutions
of a Lagrange system to be defined in this section.

\subsubsection{Notations}\label{notations:algo}

Before giving an overview of the algorithm, we need to introduce some
notations that partly extend those introduced in Subsection
\ref{sec:genericityinput}.

\paragraph*{\it Change of variables.}

We denote by $\A \circ \M$ the affine map $\X \mapsto \A(\M\X)$ obtained by
applying a change of variables with matrix $\M \in \GL(n,\CC)$.
In particular $\A = \A \circ Id_n$.

\paragraph*{\it Incidence variety.}

Given a matrix $\M \in \GL(n,\CC)$, define
\[
\begin{array}{lrcl}
\f(\A\circ\M) : & \CC^{n+m} & \to & \CC^m \\
& (\X,\Y) & \mapsto & \A(\M\X)\Y
\end{array}
\]
as a polynomial system of size $m$ in the variables $\X=(\X_1,\ldots,\X_n)$ and $\Y=(\Y_1, \ldots, \Y_m)$.
Given ${u}=(u_1, \ldots, u_m, u_{m+1})  \in \CC^{m+1}$ with $u_{m+1} \neq 0$, define
\[
\begin{array}{lrcl}
\f(\A\circ\M,u) : & \CC^{n+m} & \to & \CC^{m+1} \\
& (\X,\Y) & \mapsto &  (\A(\M\X)\Y,\: u'(y,-1))
\end{array}
\]
where $u'(y,-1) = u_1y_1 + \cdots + u_my_m-u_{m+1}$ denotes the inner product
of vectors $u$ and $(y,-1) \in \CC^{m+1}$, and let $\setV(\A\circ\M, u)=
\zeroset{\f(\A\circ\M,u)} \subset \CC^{n+m}$. We will see that under some
{\it genericity assumptions}, the algebraic variety $\setV(\A\circ\M, u)$
is equidimensional and smooth.  

\paragraph*{\it Fibers.}

Given $\fiber \in \CC$, define
\[
\begin{array}{lrcl}
\f_{\fiber}(\A\circ\M,u)  : & \CC^{n+m} & \to & \CC^{m+2} \\
& (\X,\Y) & \mapsto & (\A(\M\X)\Y,\:u'(y,-1),\:\X_1-\fiber)
\end{array}
\]
and let $\setV_{\fiber}(\A\circ\M, u) =\zeroset{\f_{\fiber}(\A\circ\M, u)} \subset \CC^{n+m}$.

We also define $\A_\fiber$ the matrix obtained by instantiating $\X_1$
to $\fiber$ in $\A$. The hypersurface defined by $\det A_\fiber =0$ is
denoted by $\setD_\fiber$.

\paragraph*{\it Lagrange system.}

Given $v \in \CC^{m+1}$, let $J(x,y)=\jac_1\f(\A\circ\M,u)$
denote the matrix of size $m+1$ by $n+m-1$ obtained by removing the first column
of the Jacobian  matrix of $\f(\A\circ\M,u)$, and define
\[
\begin{array}{lrcl}
l(\A\circ\M,u,v) : &  \CC^{n+2m+1} & \to & \CC^{n+2m+1} \\
& (\X,\Y,\Z) & \mapsto & (\A(\M\X)\Y,\:u'(y,-1), J(x,y)'\Z,\:v'z-1)
\end{array}
\]
where variables $\Z=(\Z_1, \ldots, \Z_{m+1})$ stand for Lagrange
multipliers, and
let $$\setZ(\A\circ\M, u, v) = \zeroset{l(\A\circ\M,u,v)} \subset \CC^{n+2m+1}.$$

\subsubsection{Formal description}\label{formdesc:algo}

The algorithm takes as input $\A$ which is assumed to satisfy $\sfG$.
Then, it chooses randomly $\M\in \GL(n,\QQ)$, $u \in \QQ^m$, $v \in
\QQ^{m+1}$ and $\fiber \in \QQ$ and computes a rational
parametrization of $\setZ(\A\circ\M, u, v)\subset \CC^{n+2m+1}$. Its
projection on the $(\X, \Y)$-space is expected to be the set of
critical points of the restriction to $\setV(\A\circ\M, u)$ of the
projection on the $\X_1$-coordinate. Next, a recursive call is
performed with input $\A\circ\M$ where the $\X_1$-coordinate is
instantiated to $\fiber$. The new input should satisfy the same
genericity properties as the one satisfied by $\A$.  Before giving a
detailed description of the algorithm, we describe basic subroutines
required by our algorithm. \smallskip \smallskip

\paragraph*{\it Main subroutines.} The algorithm uses the following subroutines:

\begin{itemize}
\item {\sf IsSing}: it takes as input the polynomial system $f(A \circ
  M,u)$ and it returns {\sf false} if the system satisfies $\sfG$. It
  returns {\sf true} otherwise;
\item {\sf RatPar}: it takes as input a polynomial system with coefficients in $\QQ$
  defining a finite set and it returns a rational parametrization of the set, as
defined in Section \ref{ratpar}.
\end{itemize}
It also uses the following subroutines that perform basic operations
on rational paramet\-rizations of finite sets:
\begin{itemize}
\item {\sf Image}: it takes as input a  rational parametrization of a finite set $\setZ\subset \CC^N$ and a matrix
  $\M \in \GL(N,\CC)$ and it returns a rational parametrization of the image set $\M^{-1}\setZ$
corresponding to a change of variables;
\item {\sf Union}: it takes as input two rational parametrizations of finite sets $\setZ_1,\setZ_2$
  and returns a rational parametrization of $\setZ_1\cup \setZ_2$;
\item {\sf Project}: it takes as input a rational parametrization of a finite set $\setZ$
and a subset of variables, and it computes a rational parametrization of the projection of $\setZ$
on the linear subspace generated by these variables;
\item {\sf Lift}: it takes as input a  rational parametrization of a finite set $\setZ\subset \CC^N$ and a number
$\fiber \in \CC$, and it returns a rational parametrization of $\setZ'=\{(x,\fiber) \: :\: x \in \setZ\} \subset \CC^{N+1}$.
\end{itemize}

We can now describe more precisely our algorithm {\sf RealDet}.  It
uses a recursive subroutine ${\sf RealDetRec}$ that takes as input
$\A$ satisfying $\sfG$, and returns a rational parametrization of a
finite set which meets all connected components
of $\setD \cap \RR^n$. \\

${\sf RealDetRec}(\A)$:
\begin{enumerate}
\item \label{step:1} If $n=1$ then return $(1,t,\det \A(t))$;
\item \label{step:2} Choose randomly
  \begin{itemize}
    \item $\M \in \GL(n,\QQ)$
    \item $u=(u_1,\ldots, u_{m+1}) \in \QQ^{m+1}$
    \item $v\in \QQ^{m+1}$
    \item $\fiber \in \QQ$;
  \end{itemize}
\item\label{step:3} ${\sf P}={\sf Project}({\sf RatPar}(l(\A\circ\M, u, v)), (\X_1, \ldots, \X_n))$;
\item\label{step:4} ${\sf Q}={\sf RealDetRec}({\sf Substitute}(\X_1=\fiber, \A\circ\M)))$;
\item\label{step:5} ${\sf Q}={\sf Lift}({\sf Q}, \fiber)$;
\item\label{step:6} return ${\sf Image}({\sf Union}({\sf Q}, {\sf P}), \M^{-1})$. \\
\end{enumerate}

The main algorithm {\sf RealDet} checks that the input satisfies
$\sfG$,
in which case it calls {\sf RealDetRec}. \\

${\sf RealDet}(\A)$:
\begin{enumerate}
\item \label{step:main:1} Choose randomly $u \in \QQ^{m+1}$;
\item \label{step:main:2} If ${\sf IsSing}(\f(\A, u)) = {\tt true}$ then output an error message
  saying that the genericity assumptions are not satisfied;
\item \label{step:main:3} else return ${\sf RealDetRec}(\A)$. \smallskip 
\end{enumerate}

\subsection{Proof of correctness}\label{ssec:algo:correctness}

It is immediate that it is sufficient to prove the correctness of {\sf
  RealDetRec} to obtain the correctness of {\sf RealDet}. This
algorithm takes as input an affine map $\A$ satisfying $\sfG$.

The result below shows that Assumption $\sfG$ is {\it generic} in the
sense that there a exists a non-empty Zariski open set of
$\CC^{m^2(n+1)}$ contained in the set of linear matrices satisfying
$\sfG$. It is also useful to ensure that recursive calls are valid,
i.e. the inputs in recursive calls satisfy the genericity assumption.
The proof is given in Section \ref{sec:smoothness}.

\begin{proposition}\label{prop:smoothness}
  Let $u =(u_1, \ldots, u_{m+1}) \in \QQ^{m+1}$ such that $(u_1, \ldots, u_m) \in \QQ^{m}_*$ and
  $u_{m+1} \neq 0$. Then the following holds. 
  \begin{enumerate}
  \item There exists a non-empty Zariski open set
    $\zarA\subset\CC^{m^2(n+1)}$ such that for all $\A\in\zarA$,
    $\f(\A, u)$ satisfies $\sfG$.
  \item If $\f(\A, u)$ satisfies $\sfG$ then there exists a non-empty
    Zariski open set $\zarfiber\subset \CC$ such that for any $\fiber
    \in \zarfiber$ $\f_{\fiber}(\A, u)$ satisfies $\sfG$.
  \end{enumerate}
\end{proposition}

Note that random choices are performed by algorithm {\sf
  RealDetRec} at Step \ref{step:2}. These are needed to ensure
some genericity properties. The first one ensures that set $\setZ(\A\circ\M,u,v)$ is finite; it is proved in
Section \ref{sec:dimension}.

\begin{proposition}\label{prop:dimlag}
  Assume that $\A \in \zarA $ (see Section \ref{remark:sing}) and $u \neq 0$.
  Then there exist non-empty Zariski open sets $\zarM_1 \subset \GL(n,\CC)$
  and $\zarV\subset\CC^{m+1}$ such that for all $\M \in \zarM_1 \cap
  \QQ^{m \times m}$ and $v \in \zarV \cap \QQ^{m+1}$, the following
  properties hold:
  \begin{enumerate}
  \item $\setZ(\A\circ\M, u, v)$ is a finite set;
  \item the Jacobian matrix $\jac\:l(\A\circ\M, u, v)$ has maximal rank at any point of $\setZ(\A\circ\M, u, v)$;
  \item the projection of $\setZ(\A\circ\M, u, v)$ on the $(\X,
    \Y)$-space contains the set of critical points of the restriction to $\setV(\A\circ\M, u)$
    of the projection on the $\X_1$-coordinate.
  \end{enumerate}
\end{proposition}

\noindent
The proposition below states that, for $\M\in \GL(n,\QQ)$
generically chosen, and for any connected component $\setC$ of $\setD \cap
\RR^n$, $\pi_i(\M^{-1}\setC)$ is closed for $i=1,\ldots,n-1$. This is
proved in Section \ref{sec:closedness}.

\begin{proposition}\label{prop:closedness}
  Assume that $\A \in \zarA$. Then there exist two
  non-empty Zariski open sets $\zarM_2 \subset \GL(n,\CC)$ and
  $\mathscr{U} \subset \CC^m$ such that for any
  $\M \in \zarM_2 \cap \QQ^{n \times n}$, $u \in (\zarU \cap \QQ^m) \times \QQ_*$
  and any connected component $\setC$ of $\setD \cap \RR^n$, the following
  holds:
  \begin{enumerate}
  \item for $i=1,\ldots,n-1$, $\pi_i(\M^{-1}\setC)$ is closed for the
    Euclidean topology;
  \item for any $\fiber \in \RR$ lying on the boundary of
    $\pi_1(\M^{-1}\setC)$,  $\pi_1^{-1}(\fiber) \cap \M^{-1}\setC$ is finite
    and there exists $(\X, \Y) \in \RR^n \times \RR^m$ such that
    $(\X,\Y) \in \setV(\A\circ\M, u)$ and $\pi_1(\X, \Y)=\fiber$.
  \end{enumerate}
\end{proposition}

Note that, starting with an $n$-variate affine map, there are $n$
calls to {\sf RealDetRec}, among which $n-1$  are
recursive.

The random choices performed at Step \ref{step:2} of every recursive
call to {\sf RealDetRec} can be organized in an array
\begin{equation}\label{eq:array}
\left( (\M^{(1)}, u^{(1)}, v^{(1)}, \fiber^{(1)}), \ldots, (\M^{(n-1)}, u^{(n-1)}, v^{(n-1)}, \fiber^{(n-1)}) \right)
\end{equation}
where the upperscripts indicate the depth of the recursion. There are
$n-1$ choices of these data because when $n=1$ the recursive subroutine
directly returns a rational parametrization without making such a choice.
To ensure the correctness of {\sf RealDetRec}, we need to assume that
these choices are random enough so that data $(\M^{(j)}, u^{(j)},
v^{(j)}, \fiber^{(j)})$ lie in some prescribed non-empty Zariski open
set ${\zarO}^{(j)}$ for $j=1,\ldots,n-1$ as suggested by the previous
propositions. Because of the recursive calls, {\it a priori} the set
$\zarO^{(j)}$ depends on the previous choices. This is formalized by
the following assumption. \smallskip \smallskip

{\bf Assumption} $\sfH$. We use the notations for sets introduced in 
Propositions \ref{prop:smoothness}, \ref{prop:dimlag}, \ref{prop:closedness},
with the upperscript ${(j)}$ to indicate the depth of recursion. We say
that $\sfH$ holds if the array (\ref{eq:array}) satisfies the following conditions:
\begin{itemize}
\item $\A$ satisfies assumption $\sfG$;
\item $\M^{(j)} \in \zarM_1^{(j)} \cap \zarM_2^{(j)} \cap \QQ^{n \times n}$, for $j=1,\ldots,n-1$;
\item $u^{(j)} \in (\zarU^{(j)}\times\QQ_*) \cap \QQ^{m+1}$ for $j=1,\ldots,n-1$;
\item $v^{(j)} \in \zarV^{(j)} \cap \QQ^{m+1}$, for  $j=1,\ldots,n-1$;
\item $\fiber^{(j)} \in \zarfiber^{(j)} \cap \QQ$, for  $j=1,\ldots,n-1$.
\end{itemize}

We can now prove the following correctness statement.

\begin{theorem}\label{theo:correctness}
  Assume that $\A\in\zarA$ and that
  $\sfH$ holds.  Then, ${\sf RealDet}(\A)$ returns
  a rational parametrization encoding a finite set of points 
with non-empty intersection with each connected component of $\setD\cap\RR^n$. 
\end{theorem}

\begin{proof} Our reasoning is by induction on $n$, the number of
  variables. We start with the initialization.  When $n=1$,
  $\setD\subset \CC$ is finite. Then a rational parametrization of
  $\setD$ is the triple $(1,t,\text{det} \A(t))$, which is the output result.
  Now, our induction assumption is that for any linear map
  $\X \mapsto \A(\X)=\A_0+\X_1\A_1+\cdots+\X_{n-1}\A_{n-1}$ that
  satisfies  $\sfG$, the algorithm {\sf RealDetRec}
  returns a correct answer provided that $\sfH$ holds.

  Now, let $\A$ be a linear  map and let $\setC$ be a connected
  component of $\setD\cap\RR^n$. We let  $\M, u, v$ and $\fiber$ be
  respectively the matrix, vectors and rational number chosen at Step
  \ref{step:2} of {\sf RealDetRec}, with input $\A$.

  First assume that the projection on the $\X_1$-coordinate of
  $\M^{-1}\setC$ is the whole $\X_1$-axis. Since $\A$ satisfies
  $\sfG$, we deduce that $\A\circ\M$ satisfies $\sfG$.  Since $\sfH$
  holds, we conclude by Proposition \ref{prop:smoothness} that
  $\f_{\fiber}(\A\circ\M, u)$ generates a radical ideal and defines an
  algebraic variety which is either empty or smooth
  $(n-2)$-equidimensional. Moreover the singular locus of the
  determinant obtained by instantiating $\X_1$ to $\fiber$ in
  $\A\circ\M$ is exactly the set of points at which the rank
  deficiency is greater than $1$.

  Since, by assumption, $\pi_1(\M^{-1}\setC)$ is the whole $\X_1$-axis,
  there exists a connected component $\setC'$ of the solution set of
  $\f_{\fiber}(\A\circ\M, u)$ such that $\{(\fiber, \x)\mid \x  \in \setC'\}$
 is contained in $\M^{-1}\setC$.  In other words, the input of
  {\sf RealDetRec} at Step \ref{step:4} satisfies $\sfG$ and
  it is sufficient to compute sample points in each connected
  component of the solution set of $\f_{\fiber}(\A\circ\M, u)$
  to obtain a sample point in $\setC$. Correctness follows from the
  induction assumption which implies that {\sf RealDetRec}
  computes at least one point in each connected component of the
  algebraic set defined by $\f_{\fiber}(\A\circ\M, u)$.

   Now, assume that the projection $\pi_1$ on the
  $\X_1$-coordinate of $\M^{-1}\setC$ is not the whole $\X_1$-axis.  Since
  $\sfH$ is satisfied, we deduce by Proposition \ref{prop:closedness}
  that $\pi_1(\M^{-1}\setC)$ is closed for the Euclidean topology.  Since
  $\pi_1(\M^{-1}\setC) \neq \RR$ by assumption and since $\pi_1(\M^{-1}\setC)$ is
  closed, there exists $\x=(\x_1,\ldots, \x_n)\in \M^{-1}\setC$ such that $\fiber = \x_1$ lies in the
  boundary of $\pi_1(\M^{-1}\setC)$. Without loss of generality, we assume
  below that $\pi_1(\M^{-1}\setC)$ is contained in $[\fiber, +\infty[$.

  Recall that $\sfH$ holds. Then, by Proposition
  \ref{prop:closedness}, $\pi_1^{-1}(\fiber) \cap \M^{-1}\setC$ is finite and
  for all $\x \in \pi_1^{-1}(\fiber)\cap \M^{-1}\setC$ there exists $\y\in \RR^m$
  such that $(\x, \y)\in \setV(\A\circ\M, u)$.

  Below, we reuse the notations of the algorithm and we prove that
  there exists $\z \in \CC^{m+1}$ such that $(\x, \y, \z)$ is a point
  lying in $\setZ(\A\circ\M, u, v)$.  Combined with Proposition
  \ref{prop:dimlag}, we also deduce that the above polynomial system
  defines a finite set which contains $(\x, \y, \z)$.  Thus, the calls
  to {\sf RatPar} and {\sf Project} are valid and
  $(\x, \y, \z)$ lies in the finite set of points computed at Step
  \ref{step:3} of {\sf RealDetRec}. Correctness of the
  algorithm follows straightforwardly.

  Thus, it remains to prove that there exists $\z\in \CC^{m+1}$ such
  that $(\x, \y,\z)$ lies in $\setZ(\A\circ\M, u, v)$.  Let
  $\M^{-1}\setC'$ be the connected component of $\setV(\A\circ\M,
  u)\cap\RR^{n+m}$ which contains $(\x, \y)$. We claim that
  $\fiber=\pi_1(\x, \y)$ lies on the boundary of $\pi_1(\M^{-1}\setC')$.

  Indeed, assume by contradiction that this is not the case, i.e.
  $\fiber \in \pi_1(\M^{-1}\setC')$ but does not lie in the boundary of
  $\pi_1(\M^{-1}\setC')$. This implies that there exists $\varepsilon >0$
  such that the interval $(\fiber-\varepsilon, \fiber+\varepsilon)$
  lies in $\pi_1(\M^{-1}\setC')$. As a consequence, there exists $(\x',
  \y')\in\M^{-1}\setC'$ such that that $\pi_1(\x',\y')<
  \fiber$. Moreover, since $(\x',\y')\in \M^{-1}\setC'$ and $\M^{-1}\setC'$ is
  connected, we deduce that there exists a continuous semi-algebraic
  function $\tau \colon [0, 1]\to \M^{-1}{\setC'}$ with $\tau(0)=(\x,
  \y)$ and $\tau(1)=(\x',\y')$. Let $\pi_\X$ be the projection map
  $\pi_\X(\X,\Y) = \X$. Since $\pi_\X$ and $\tau$ are
  continuous semi-algebraic functions, $\gamma = \pi_\X \circ \tau$
  is continuous and semi-algebraic (since it is the composition of
  semi-algebraic continuous functions). Finally, note that
  $\gamma(0)=\x$ and $\gamma(1)=\x'$; we deduce that $\x'\in \M^{-1}\setC$
  with $\pi_1(\x')< \fiber=\pi_1(\x)$. This contradicts the fact that
  $\fiber$ lies in the boundary of $\pi_1(\M^{-1}\setC)$ and that
  $\pi_1(\M^{-1}\setC)$ lies in $[\fiber, +\infty[$. We conclude that
  $\fiber=\pi_1(\x, \y)$ lies on the boundary of $\pi_1(\M^{-1}\setC')$.

  As a consequence of the implicit function theorem \cite[Section
  3.5]{BaPoRo06}, we deduce that $(\x, \y)$ is a critical point of the
  restriction of the projection $\pi_1$ to $\M^{-1}\setC'$. Since
  $\M^{-1}\setC'$ is a connected component of $\setV(\A\circ\M, u)\cap
  \RR^n$ and since the input satisfies $\sfG_2$, we deduce that the
  truncated Jacobian matrix $\jac_1\f(\A\circ\M,u)$ (defined jointly
  with the Lagrange system in paragraph \ref{notations:algo}) is rank
  defective at $(\x, \y)$ (see \cite[Sections 2.1.4 and
  2.1.5]{SaSc13}).  Moreover, since $\sfH$ holds, we deduce by
  Proposition \ref{prop:dimlag} that $(\x, \y)$ lies in the projection
  of $\setZ(\A\circ\M, u, v)$. Thus, there exists $\z\in \CC^{m+1}$
  such that $(\x, \y,\z)$ lies in $\setZ(\A\circ\M, u, v)$ as
  requested.
\end{proof}

\subsection{Complexity analysis and degree bounds}\label{ssec:algo:complexity}

In this section, we estimate the complexity of the algorithm {\sf
  RealDet} and we give an explicit formula for a bound on the number
of complex solutions computed by the algorithm.

We assume that $\sfG$ holds, so that we
do not need to estimate the complexity of the subroutine {\sf IsSing}
and we focus on the complexity of the algorithm {\sf RealDetRec}. 

We assume in the sequel that $\sfH$ holds. On input $\A$
satisfying $\sfG$, {\sf RealDetRec} computes a
rational parametrization of the solutions set of $l(\A\circ\M, u, v)$
 (Step \ref{step:3}) and performs a
recursive call with input ${\sf Substitute}(\X_1 = \fiber,
\A\circ\M)$ (Step \ref{step:4}). On input $l(\A\circ\M, u, v)$,
our routine for computing rational parametrization of
its solution set starts by building an equivalent system.

The complexity results stated below depend on  degrees of
geometric objects defined by systems which are equivalent to the
Lagrange systems we consider.  We need to introduce some notations.

  The sequence of linear matrices considered during the
  recursive calls is denoted by $\A^{(0)}, \ldots, \A^{(n-1)}$, where
  $\A^{(i)}$ is a linear matrix in $n-i$ variables; the
  systems $\f(\A^{(i)}\circ\M^{(i)}, u^{(i)})$ for $0 \leq i \leq n-1$ are
  respectively denoted by
  \[
  \f_i=(f_{i,1}, \ldots, f_{i,m+1})
   \]
  where $f_{i,m+1} : \Y \mapsto \Y'u^{(i)}-1$.
  Note that the $\f_i$ involve $n+m-i$ variables.  The Lagrange systems
  $l(\A^{(i)}\circ\M^{(i)}, u^{(i)},v^{(i)})$ are denoted by
  \[
  l_i=(\f_i, g_{i, 1}, \ldots, g_{i, n+m-i})
  \]
  where $g_{i, n+m-i} : \Z \mapsto v^{(i)}\Z_{n+m-i+1}-1$.

  Using $f_{i,m+1}$, one can eliminate one of the $\Y$-variables, say
  $\Y_m$, in $\f_i$. We denote by
  \[
  \tilde \f_i =(\tilde f_{i, 1}, \ldots, \tilde f_{i, m})
  \]
  the polynomial system obtained this way. Recall that the polynomials
  $\g_{i, 1}, \ldots, \g_{i, n+m-i}$
  express that there is a nonzero vector in the left kernel of the
  truncated Jacobian matrix $\jac_1\f_i$. Hence, one can equivalently express
  the existence of a non-zero vector in the left kernel of the truncated
  Jacobian matrix $\jac_1\tilde\f_i$. This yields a new polynomial system
  \[
  \tilde l_i = (\tilde f_{i,1}, \ldots, \tilde f_{i, m}, \tilde g_{i, m+1}, \ldots,
  \tilde g_{i, n-i-1}, \tilde g_{i, n-i}, \ldots, \tilde g_{i,
    n+2m-i-2}).
  \]
Note that since we have assumed that $\sfH$ holds, one can deduce
using Proposition \ref{prop:dimlag} that the Jacobian matrix $\jac\tilde l_i$
has maximal rank at any complex solution of $\tilde l_i$.

This new polynomial system contains:
\begin{itemize}
\item $m$ polynomials which are bilinear in $(\X_1, \ldots, \X_n)$
  and $(\Y_1, \ldots, \Y_{m-1})$;
\item $m-1$ polynomials which are bilinear in $(\X_1, \ldots, \X_n)$
  and $(\Z_1, \ldots, \Z_{m-1})$;
\item $n-1$ polynomials which are bilinear in $(\Y_1, \ldots, \Y_{m-1})$
  and $(\Z_1, \ldots, \Z_{m-1})$. 
\end{itemize}
In the sequel, we denote by $\setV_{i,j}$ the
algebraic set defined by 
\[
\tilde f_{i,1}, \ldots, \tilde f_{i, j}, \qquad \text{ when } 1\leq j \leq m
\]
and 
\[
\tilde f_{i,1}, \ldots, \tilde f_{i, m}, \tilde g_{i, m+1}, \ldots, \tilde g_{i, j} \qquad \text{ when } m+1\leq j \leq n+2m-i-2.
\]
The algebraic set $\setW_{i,j}$ is the subset of $\setV_{i,j}$ at which the Jacobian
matrix of its above defining system has maximal rank. 
For $0 \leq i \leq n-1$, we denote by
$$
\delta_i = \max \{\deg \setW_{i,j} \: :\: 1\leq j \leq n+2m-i-2\}
$$
and by $\delta$ the maximum of the $\delta_i$. Remark that since $\sfH$ holds, Proposition
\ref{prop:dimlag} implies that $\setW_{i, n+2m-i-2}=\zeroset{\tilde l_i}$.

We start by estimating the complexity of the main subroutines called
by {\sf RealDetRec}. We prove the following result.

\begin{proposition}\label{prop:complexity:RealDet}
  Assume that $\sfH$ holds. Then, {\sf RealDetRec} outputs a rational
  parametrization whose real zero locus meets each connected component
  of $\setD\cap \RR^n$ within
  $$ \softO
  \left(n^2m^2(n+m)^5 {\delta}^2\right) $$ arithmetic operations in
  $\QQ$ with $\delta \leq {{n+m}\choose{m}}^3$ and $\softO(s)={O}(s\cdot\log^k(s))$ for some $k \in \mathbb{N}$.
\end{proposition}

Assume that $\A$ satisfies $\sfG$ and that $\M$, $u$ and $\v$ lies
in the non-empty Zariski open sets defined in Propositions
\ref{prop:dimlag} and \ref{prop:closedness}.

\begin{lemma}\label{lemma:complexity:lagrange}
  Under the above notations and assumptions, there exists a
  probabilistic algorithm which, on input $l_i$, computes a
  rational parametrization of the complex solution set of it within
$$\softO\left (n^2m^2(n+m)^5\delta^2\right )$$
arithmetic operations in $\QQ$ with $\delta \leq
{{n+m}\choose{m}}^3$.
\end{lemma}

\begin{proof}[of Proposition \ref{prop:complexity:RealDet}]
  Through its recursive calls, the algorithm {\sf RealDetRec}
  computes rational parametrizations of the solution sets of the Lagrange
  systems $l_0, \ldots, l_{n-1}$. 

  Lemma \ref{lemma:complexity:lagrange} shows that these computations
  are done within
  \[
  \softO\left ((n+m)^2(nm^2+(n+m)^3)\delta^2\right )
  \]
  arithmetic operations in $\QQ$ with $\delta \leq
  {{n+m}\choose{m}}^3$. Since there are $n$ Lagrange systems to solve,
  all these parametrizations are computed within
  $$ \softO
  \left(n^2m^2(n+m)^5 {\delta}^2\right) $$ arithmetic operations in
  $\QQ$. Note that in all systems $l_0, \ldots, l_{n-1}$ the number of variables
  is bounded by $n+2m+1$ and
  the cardinality of their solution set is bounded by $\delta$.

  Following \cite[Lemma 10.1.3]{SaSc13}, the call to the routine {\sf
    Project} at Step \ref{step:3} requires at most
  $\softO((n+m)\delta^2)$ arithmetic operations in $\QQ$.

  Next, by \cite[Lemma 10.1.1 and Lemma 10.1.3]{SaSc13}, the calls to
  the routines {\sf Image} and {\sf Union} and in
  Step \ref{step:6} require respectively at most
  $\softO((n+m)^2\delta+(n+m)^3)$ and $\softO((n+m)\delta^2)$ arithmetic
  operations in $\QQ$.  Summing up all these complexity estimates
  yields to the announced complexity bounds.
\end{proof}

\begin{proof}[of Lemma \ref{lemma:complexity:lagrange}]
  It is sufficient to describe the proof for $l=l_0$ only.
  We use the geometric resolution algorithm given in \cite{GiLeSa01}
  to compute a rational parametrization of the complex solution set of
  the system $\tilde l$ obtained following the construction
  in Paragraph \ref{ssec:algo:complexity}. Note that since $\sfH$
  holds by assumption, we deduce that $\tilde l$ is a reduced
  regular system, in the sense defined in the introduction of \cite{GiLeSa01}.

  Note that all polynomials of $\tilde l$ have degree $\leq
  2$ and that evaluating $\tilde l$ requires $\softO(nm^2)$
  arithmetic operations.

  Thus, one can apply \cite[Theorem 1]{GiLeSa01}. When $\tilde l$
   is a { reduced regular sequence}, it states that one can
  compute a rational parametrization of the complex solution set of
  $\tilde l$ in probabilistic time $$\softO(\tilde{n}^2(\tilde{o}+\tilde{n}^3)\delta^2)$$
where
\begin{itemize}
\item $\tilde{n}=n+2m-2$ is the total number of variables involved in $\tilde l$, 
\item $\tilde{o}$ is the complexity of evaluating $\tilde l$, 
\item and $\delta$ is the quantity introduced in Paragraph \ref{ssec:algo:complexity}.
\end{itemize}
We obtain that one can compute a rational parametrization of the
complex solution set of $\tilde l$ in probabilistic time
$$\softO((n+m)^2(nm^2+(n+m)^3)\delta^2).$$
Our conclusion follows and the bound on $\delta$ is proved in the
following lemma.
\end{proof}

\begin{lemma}\label{lemma:compl:RUR}
Under the above notations and assumptions the following inequality holds: 
$$
\delta\leq {{n+m}\choose {m}}^3. 
$$
\end{lemma}

\begin{proof}
  To prove degree bounds on $\delta$, we take into account the
  multi-linear structure in $\X, \Y, \Z$ of the intermediate systems
\[
\tilde f_{i,1}, \ldots, \tilde f_{i, t}, \qquad \text{ for } 1\leq j \leq t
\]
and 
\[
\tilde f_{i,1}, \ldots, \tilde f_{i, m}, \tilde g_{i, m+1}, \ldots, \tilde g_{i, m+t} \qquad \text{ for } 1\leq t \leq n+2m-i-2.
\]
We define $\Delta(m,n;t)$ as follows:
\begin{itemize}
\item when $1\leq t\leq m$, $\Delta(m,n;t)$ is the sum of the coefficients of the
  polynomial $(s_1+s_2)^{t}$ modulo the ideal generated by $(s_1^{n+1}, s_2^{m})$;
\item when $m+1\leq t\leq n+m-1$, $\Delta(m,n;t)$ is the sum of the coefficients of the
  polynomial $(s_1+s_2)^{m}(s_1+s_3)^{t-m}$ modulo the ideal generated by $(s_1^{n+1}, s_2^{m}, s_3^{m})$; 
\item when $n+m\leq t\leq n+2m-2$, $\Delta(m,n;t)$ is the sum of the coefficients of the
  polynomial $(s_1+s_2)^{m}(s_1+s_3)^{n-1} (s_3+s_2)^{t-m-n+1}$ modulo the ideal
generated by $(s_1^{n+1}, s_2^{m}, s_3^{m})$. 
\end{itemize}

By \cite[Proposition 10.1.1]{SaSc13}, the degrees of {\it their
  components of highest dimension} is bounded by $\Delta(m,n;t)$.
Immediate computations show that the following holds:
$$
\Delta(m,n;t)=
\begin{cases}
  \sum_{i=0}^{\min(n,t)}\binom{t}{i} & t\in\{1,\dots,m\}, \\
  \sum_{(i,j)\in\mathcal{F}_t}\binom{m}{i}\binom{t-m}{j} & t\in\{m+1,\dots,n+m-1\}, \\
  \sum_{(i,j,\ell)\in\mathcal{F}_t}\binom{m}{i}\binom{n-1}{j}\binom{t-m-n+1}{\ell} & t\in\{n+m,\dots,n+2m-2\}.
\end{cases}
$$
for every $m$ and $n$, where:
$$
\mathcal{F}_t= \ \ 
 \begin{cases}
  (i,j) \in \{1,\dots,m\} \times \{0,\dots,n-1\}, \\
  1 \leq i \leq \min(m,n), \\
  \max(0,t-2m+1) \leq j \leq \min(t-m,i-1),
 \end{cases}
$$
if $t\in\{m+1,\dots,n+m-1\}$, and
$$
\mathcal{F}_t= \ \ 
 \begin{cases}
  (i,j,\ell) \in \{1,\dots,m\} \times \{0,\dots,n-1\} \times \{0,\dots,t-m-n+1\}, \\
  \max(0,t-2m+1) \leq j+\ell \leq n-1,\\
\max(1,t-2m+2) \leq i+\ell \leq \min(n,t-n+1).
 \end{cases}
$$
if $t\in\{n+m,\dots,n+2m-2\}$. Let us remark that relations defining $\mathcal{F}_{n+2m-2}$ become
 linear contraints, which yields the following equality for the case $t=n+2m-2$:
\begin{equation}\label{bound-complex}
\Delta(m,n;n+2m-2)=\sum_{i=0}^{m-1}\binom{m}{n-i}\binom{n-1}{i}\binom{m-1}{i}.
\end{equation}

One can easily check that for all $k \in \mathbb{N}$
$$
\binom{n+m}{n}^k = \sum_{i_1, \ldots, i_{k}=0}^{n}{\binom{m}{i_1}\binom{n}{i_1}\cdots\binom{m}{i_k}\binom{n}{i_k}}.
$$
Moreover, for all $m,n$ and for $t \in \{1, \ldots, m\}$,
$\Delta(m,n;t) \leq \Delta(m,n;t+1)$, and $\Delta(m,n;m)$ is bounded
by $\binom{n+m}{n}$ because of the previous formula.

Let $t \in \{m+1,\dots,n+m-1\}$. Then $\Delta(m,n;t) = \sum_{i=1}^{\min(m,n)}{a_i\binom{m}{i}}$ where
$$
a_i=\sum_{j:(i,j)\in\mathcal{F}_t}{\binom{t-m}{j}}=\sum_{j=\max(0,t-2m+1)}^{\min(t-m,i-1)}{\binom{t-m}{j}} \leq \sum_{j=0}^{n}{\binom{n}{i}\binom{m}{j}\binom{n}{j}}.
$$
and so $\Delta(m,n;t) \leq \binom{n+m}{n}^2$ for all $t \in \{m+1,\dots,n+m-1\}$.

Finally, for $t \in \{n+m,\dots,n+2m-2\}$, one gets
 $$\begin{array}{rcl}
 \Delta(m,n;t) &\leq &\sum_{i,j,\ell=0}^n\binom{m}{i}\binom{n-1}{j}\binom{t-m-n+1}{\ell}\\
& \leq &\sum_{i,j,\ell=0}^n \binom{m}{i}\binom{n}{j}\binom{m}{\ell}\\
& \leq &\binom{n+m}{n}^3.
 \end{array}
 $$
\end{proof}

\subsubsection{Complexity of {\sf Project}.}\label{subsub:comp}

According to \cite[Lemma 9.1.6]{SaSc13},
given a rational parametrization $\q$ defining a
zero-dimensional set $\setZ \subset \CC^N$, there exists a probabilistic
algorithm computing a rational parametrization $\q'$ of
the projection $\pi_i(\setZ)$ whose complexity is within
$\mathcal{O}^{\thicksim}(N\deg\:\q^2)$ operations.  We remark here
that $\deg\:\q$ is the cardinality of $\setZ$ provided that
$\q$ is square-free; if not, it is an upper bound. In the case of
$\setZ(\A\circ\M,u,v)$, we obtain from Lemma
\ref{lemma:compl:RUR} that $\deg\:\q \leq \binom{n+m}{n}^3$.

\begin{lemma}\label{lemma:compl:proj}
The complexity of {\sf Project} in ${\sf RealDetRec}$ is
$$
\mathcal{O}^{\thicksim}
\left(
(n+2m-2) \binom{n+m}{n}^6
\right).
$$
\begin{proof}
It follows from the bound for $\delta$ of Lemma \ref{lemma:compl:RUR} and from \cite[Lemma 9.1.6]{SaSc13}.
\end{proof}
\end{lemma}

\subsubsection{Complexity of {\sf Image, Union}.}\label{subsub:others}

By \cite[Lemma 9.1.1]{SaSc13}, given a rational parametrization $\q$
and a matrix $\M\in \GL(N,\QQ)$, there exists an algorithm
computing a rational parametrization $\q'$ such that $\setZ(\q')=\M^{-1}\setZ(\q)$ using
$\mathcal{O}^{\thicksim}(N^2\delta+N^3)$ operations. Moreover, by \cite[Lemma 9.1.3]{SaSc13}
if $\q_1,\q_2$ are rational parametrizations with degree sum bounded by $\delta$,
a rational parametrization of $\setZ(\q_1) \cup \setZ(\q_2)$ can be computed in $\mathcal{O}^{\thicksim}(N\delta^2)$
operations.

\begin{lemma}\label{lemma:compl:others}
The complexity of {\sf Image} and {\sf Union} in ${\sf RealDetRec}$ is
$$
\mathcal{O}^{\thicksim}
\left(
(n+2m-2)^2\binom{n+m}{n}^3+(n+2m-2)^3+(n+2m-2)\binom{n+m}{n}^6
\right).
$$
\begin{proof}
The proof of this fact follows straightforwardly from \cite[Lemma 9.1.1]{SaSc13},
\cite[Lemma 9.1.3]{SaSc13} and Lemma \ref{lemma:compl:RUR}.
\end{proof}
\end{lemma}

\subsubsection{A bound on the degree of the output}\label{subsub:proj}
Let $A$ be a $n-$variate linear matrix of size $m$ and apply algorithm
${\sf RealDet}$ to $A$. Recall that the number $\Delta(m,n,n+2m-2)$ computed
in \eqref{bound-complex} is a bound on the number of complex solutions computed
by the first call of ${\sf RealDetRec}$.

The following result, whose proof is straightforward, counts the maximum number
of complex solutions computed by ${\sf RealDet}$.

\begin{lemma}
The number of complex solutions computed by ${\sf RealDet}$ with input a
linear matrix $A$ satisfying Assumption ${\sfG}$, is upper-bounded by the number
$$
b(m,n) = \sum_{j=1}^n\Delta(m,j,j+2m-2) = \sum_{j=1}^n\sum_{i=0}^{m-1}\binom{m}{j-i}\binom{j-1}{i}\binom{m-1}{i}.
$$
\end{lemma}
We remark the following facts:
\begin{itemize}
\item $\Delta(m,j,j+2m-2)=0$ if $j \geq 2m$;
\item if $m=m_0$ is fixed, $n \mapsto b(m_0,n)$ is constant if $n \geq 2m_0$.
\end{itemize}

\section{Regularity properties of the incidence variety}\label{sec:smoothness}

The aim of this section is to prove Proposition \ref{prop:smoothness}.
We start with $\sfG_1$. 

Indeed, the genericity of $\sfG_1$ is already established in the proof
of \cite[Theorem 5.50]{bpt13} using Bertini's theorem (see Subsection
\ref{remark:sing}). We denote by $\zarA_1 \subset \CC^{m^2(n+1)}$ a
non-empty Zariski open set such that for any $\A\in \zarA$, $\A$
satisfies $\sfG_1$.

Moreover, Sard's theorem implies that for a generic $\fiber$, a point
in $\setD_\fiber$ is regular if and only if it is regular in
$\setD$. We deduce that at any singular point of $\setD_\fiber$ the
rank deficiency of $A_\fiber$ is greater than $1$. We denote by
$\zarfiber_1$ a non-empty Zariski open set such that for any $\fiber
\in \zarfiber$, $\A_\fiber$ satisfies $\sfG_1$.

We focus on $\sfG_2$ which requires more technical proofs. 
To identify the linear map $\X\mapsto\A(\X)=\A_0+\X_1\A_1+\cdots+\X_n\A_n$
with a point in $\CC^{m^2(n+1)}$, we
denote by $\fraka_{l,i,j}$ the entry of matrix $\A_l$ at row $i$ and column $j$,
for $l=0,1,\ldots,n$ and $i,j=1,\ldots,m$.

\begin{proof}[of the first point of Proposition \ref{prop:smoothness}]
Consider the polynomial map
\[
  \begin{array}{lrcc}
  p : &  \CC^{n+m}\times\CC^{m^2(n+1)} & \longrightarrow & \CC^{m+1} \\
            &  (\x,\y,\frakA) & \longmapsto& \f(\frakA, u)
  \end{array}
\]
and, for a given $\A\in \CC^{m^2(n+1)}$, the induced map  
\[
  \begin{array}{lrcc}
  p_\A: &  \CC^{n+m} & \longrightarrow & \CC^{m+1} \\
            &  (\x,\y) & \longmapsto& p(\x,\y,\A).
  \end{array}
\]

Assume first that $\zeroset{p}$ is empty. This is equivalent to saying
that, for any $\A \in \CC^{m^2(n+1)}$, $\setV(\A,
u)=\zeroset{\f(\frakA,u)}$ is empty. By the Nullstellensatz
\cite[Chap. 8]{CLO}, this implies that for any $\A \in
\CC^{m^2(n+1)}$, the ideal $\ideal{\setV(\A,u)} = \langle \f(\frakA,u)
\rangle = \langle 1 \rangle$ is radical. We define $\zarA_2 =
\CC^{(n+1)m^2}$ and conclude by taking $\zarA=\zarA_1\cap \zarA_2$.

Assume that $\zeroset{p}$ is non-empty. We
prove below that there exists a non-empty Zariski open set $\zarA_2
\subset\CC^{m^2(n+1)}$ such that for any $\A \in \zarA_2$, the Jacobian
matrix $\jac\f(\A, u)$ has maximal rank at any
point in $\zeroset{p_{\A}}$. This is sufficient to establish the
requested property $\sfG_2$ since by the Jacobian criterion
\cite[Theorem 16.19]{Eisenbud95} this implies that 

\begin{itemize}
\item the ideal $\langle \f(\A, u) \rangle$ is radical; 
\item the algebraic set $\setV(\A, u)$ is either empty
or smooth and equidimensional of co-dimen\-sion $m+1$ in $\CC^{n+m}$. 
\end{itemize}

To prove the existence of the aforementioned non-empty Zariski open
set $\zarA_2$, we first need to prove that $0$ is a regular value
of $p$, i.e. at any point of the fiber $\zeroset{p}$ the
Jacobian matrix $\jac p$ with respect to variables $\X$, $\Y$ and ${\fraka}_{\ell, i,j}$
has maximal rank.  Take $(\x, \y, \A) \in \zeroset{p}$. It suffices to prove
that there exists a maximal minor of $\jac\f(\frakA,u)$
which is not zero at $(\x, \y, \A)$.

Remark that, since $\y$ is a solution of the equation $u'(y, -1) = 0$
and $u_{m+1} \neq 0$, there exists $1 \leq s \leq m$ such that
$\y_s \neq 0$. Moreover, since $(u_1, \ldots, u_m) \neq 0$, there exists
$1 \leq \ell \leq m$ such that $u_\ell\neq 0$. Now consider the
submatrix of $Df(A, u)$ obtained by selecting
\begin{itemize}
\item the partial derivatives with respect to $\Y_\ell$ where $\ell$ is as above;
\item the partial derivatives with respect to ${\fraka}_{0, r, s}$ for all $1 \leq r \leq m$
and for $s$ as above.
\end{itemize}
Checking that this submatrix has maximal rank at $(\x, \y, \A)$ is
straightforward since
\begin{itemize}
\item the partial derivatives of the entries of ${\frakA}(\X)\Y$ with
  respect to ${\fraka}_{0, r, s}$ for $1\leq r \leq m$ is the
  diagonal matrix with $\y_s \neq 0$ on the diagonal;
\item the partial derivative of the polynomial $u'(y,-1)$
  with respect to $\Y_\ell$ is $u_\ell\neq 0$,
  while the partial derivatives of that polynomial with respect to
  ${\fraka}_{0, r, s}$ are $0$.
\end{itemize}
Thus, up to reordering the columns of this submatrix, it is triangular
with non-zero entries on the diagonal.  Finally, we conclude that $0$
is a regular value of $p$. By Thom's Algebraic Weak Transversality
theorem \cite[Section 4.2]{SaSc13} there exists a non-empty Zariski
open set $\zarA_2 \subset \CC^{m^2(n+1)}$ such that, for every $\A \in
\zarA_2$, $0$ is a regular value of the map $p_{\A}$. Taking again
$\zarA=\zarA_1\cap\zarA_2$ concludes the proof.
\end{proof}

\begin{proof}[of the second point of Proposition \ref{prop:smoothness}]
Let $\A \in \zarA$ and consider the map
\[
  \begin{array}{lrcc}
  \pi_1 : &  \setV(\A, u) & \to & \CC \\
            &  (\x,\y) & \mapsto & \x_1.
  \end{array}
\]
which is the restriction  to $\setV(\A, u)$ of the projection on the first variable.
Since $\A \in \zarA$, the variety $\setV(\A, u)$ is
either empty or smooth and equidimensional and by Sard's Lemma (\cite[Section
4.2]{SaSc13}) the image by $\pi_1$ of the set of critical points of
$\pi_1$ is contained in an algebraic hypersurface of $\CC$ (that is, a finite set).
This implies that there exists a non-empty Zariski open set $\zarfiber_2 \subset \CC$
such that if $\fiber_2 \in \zarfiber$, at least one of the following
fact holds:
\begin{itemize}
\item the set $\pi_1^{-1}(\fiber) = \{(\x, \y) \in \setV(\A, u)
  \mid \x_1 = \fiber\}$ is empty: this fact implies that the system $\f_{\fiber}(\A, u)$
  defines the empty set, and that $\left \langle \f_\fiber(\A, u)\rangle
    = \langle 1 \right \rangle $, which is a radical ideal;
\item for all $(\x, \y) \in \pi_1^{-1}(\fiber)$, $(\x, \y)$ is not a critical point
  of $\pi_1$; this fact implies that the Jacobian matrix of $\f_\fiber(\A, u)$
  has full rank at each point $(\x, \y)$ in the zero set of $\f_\fiber(\A, u)$,
  and so by the Jacobian criterion \cite[Theorem 16.19]{Eisenbud95} that
  $\f_\fiber(\A, u)$ defines a radical ideal and its zero set
  is a smooth equidimensional algebraic set of codimension $m+2$ in $\CC^{n+m}$.
\end{itemize}
By this, we conclude that if $\fiber \in \zarfiber_2$, the system
$\f_{\fiber}(\A, u)$ satisfies $\sfG_2$ as requested. We conclude by
taking $\zarfiber=\zarfiber_1\cap \zarfiber_2$.
\end{proof}

\begin{example}
Consider the linear matrix
$$
\A(\X)=
\left(\begin{array}{ccc}
1    & \X_1 & \X_2 \\
\X_1 &    1 & \X_3 \\
\X_2 & \X_3 &    1
\end{array}\right)
$$
whose real determinantal variety is the Cayley cubic surface
with its four singular points $(x_1,x_2,x_3) \in \{(1,1,1),(1,-1,-1),(-1,1,-1),(-1,-1,1)\}$,
see Example 2 and Figure 3 in \cite{sturmfelsnieranestad}.
When evaluated at these points, $\A$ has rank one.
The following \texttt{Macaulay2} code shows that the incidence variety is smooth. \smallskip

\begin{verbatim}
MyRand = () -> (((-1)^(random(ZZ)))*(random(QQ)))
R = QQ[x_1,x_2,x_3]
A = matrix{{1,x_1,x_2},{x_1,1,x_3},{x_2,x_3,1}}
D = ideal det A
dim D, degree D
SingD = ideal singularLocus D
dim SingD, degree SingD
S = QQ[x_1,x_2,x_3,y_1,y_2,y_3]
Y = matrix{{y_1},{y_2},{y_3}}
V = ideal(sub(A,S)*Y) + ideal(1-sum(3,i->MyRand()*(y_(i+1))))
dim V, degree V
SingV = ideal singularLocus V
dim SingV, degree SingV
\end{verbatim}

\noindent The incidence variety in this example has dimension $2$ and degree $7$.

\end{example}

\begin{figure}[!ht]
\centering
\includegraphics[width=0.7\textwidth]{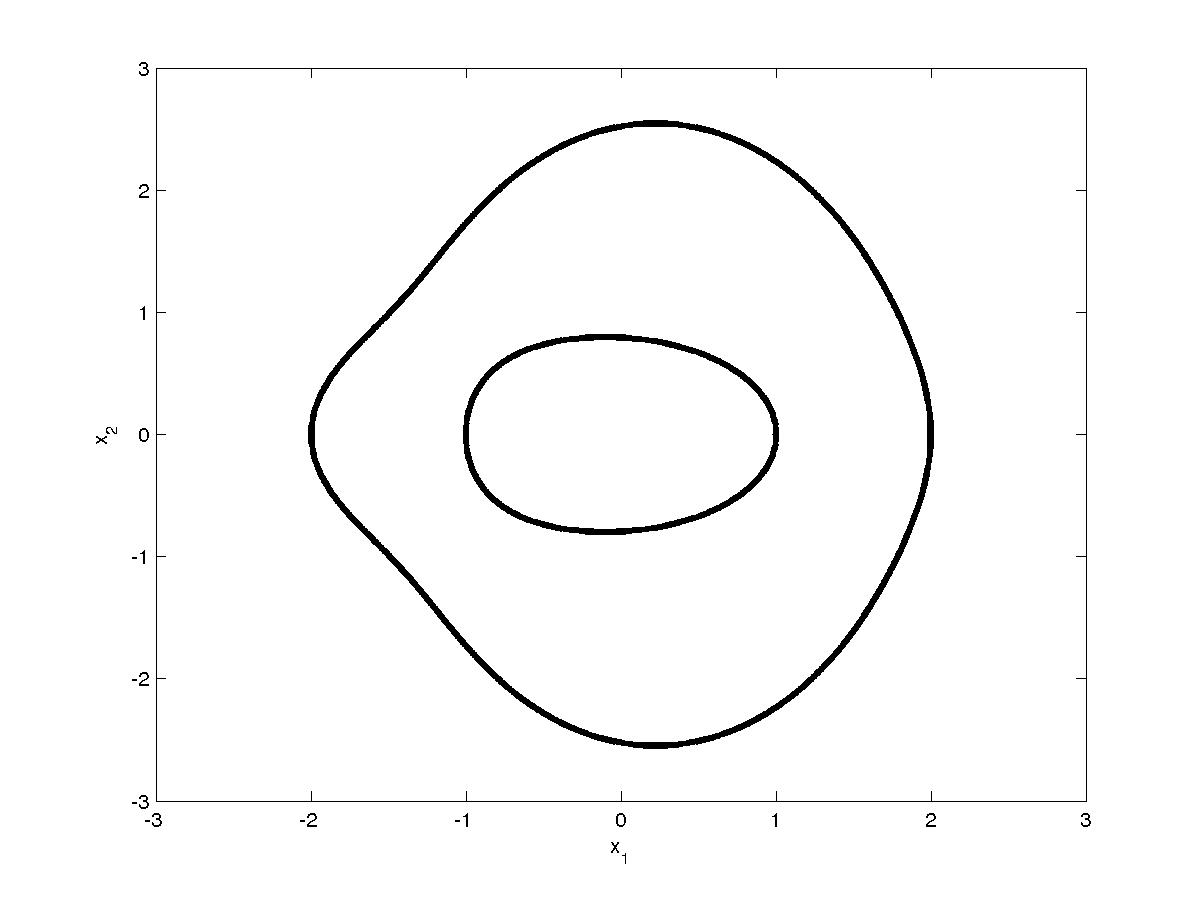}
\caption{The smooth quartic curve of Example \ref{ex:quartic} with its two nested ovals.\label{fig:quartic}}
\end{figure}      

\begin{example}\label{ex:quartic}
Consider the linear matrix
$$
\A(\X)=
\left(\begin{array}{cccc}
1+\X_1  & \X_2   & 0 & 0 \\
  \X_2  & 1-\X_1 & \X_2 & 0 \\
     0  & \X_2   & 2+\X_1 & \X_2 \\
     0  & 0      & \X_2 & 2-\X_1
\end{array}\right)
$$
whose real determinantal variety is a smooth quartic curve, the union of two
nested ovals, see Figure \ref{fig:quartic}.
Here the incidence variety is a smooth variety of dimension  $6$ and degree $10$.   
\end{example}

\section{Dimension properties of Lagrange systems} \label{sec:dimension}

The goal of this section is to prove Proposition \ref{prop:dimlag}.
First we provide a local description of the incidence variety $\setV(A
\circ M, u)$ depending on the rank of the matrix $A(x)$, which induces
a local description of the Lagrange system $\setZ(A \circ M, u, v)$.
We actually use this local description to prove that, generically, the
solutions $(x, y, z)$ of the Lagrange system are such that the rank of
$A(x)$ is $m-1$. This property is used to exploit a local description
of the solutions to the Lagrange system which make easier the proof of 
Proposition \ref{prop:dimlag}. 

Then we give a proof of it in Section \ref{subsec:proof-prop}. Recall
that the hypotheses of Proposition \ref{prop:dimlag} includes the fact
that $A \in \zarA$, where these Zariski open sets
have been defined respectively in Proposition \ref{prop:smoothness}
and Section \ref{remark:sing}.

In the sequel, for $f\in \QQ[x]$, we denote by $\QQ[x]_f$ the localized
polynomial ring.

\subsection{Local description} \label{subsec:local-des}

Let $A(x) = A_0+x_1A_1+\ldots+x_nA_n \in \QQ[x]^{m \times m}$ be a linear matrix
on $(x_1, \ldots, x_n)$, and suppose that $N=N(x) \in \QQ[x]^{r \times r}$ is
the upper-left $r \times r$ submatrix of $A$:
$$
A = 
\left[
\begin{array}{cc}
  N  & Q \\
  P' & R \\
\end{array}
\right]
$$
with $Q \in \QQ[x]^{r \times (m-r)}$, $P' \in \QQ[x]^{(m-r) \times r}$ and $R \in \QQ[x]^{(m-r) \times (m-r)}$.
Consider the equations $Ay=0$ and $u'y-u_{m+1}=0$ defining the incidence variety. 

\begin{lemma} \label{lemma:local} Let $A,N,Q,P,R$ be as above, and $u
  \neq 0$. Then there exist $\{q_{i,j}\}_{1 \leq i \leq r, 1 \leq j
    \leq m-r} \subset \QQ[x]_{\det N}$ and $\{q'_{i,j}\}_{1 \leq i,j
    \leq m-r} \subset \QQ[x]_{\det N}$ such that the constructible set
  $\setV \cap \{(x,y) : \det N \neq 0\}$ is defined by the equations
\begin{align*}
y_i-q_{i,1}(x)y_{r+1}- \ldots - q_{i,m-r}(x)y_m &= 0 \qquad i=1, \ldots, r \\
q'_{i,1}(x)y_{r+1} + \ldots + q'_{i,m-r}(x)y_{m} &= 0 \qquad i=1, \ldots, m-r \\
u'y-u_{m+1} &= 0.
\end{align*}
\end{lemma}

\begin{proof}
Consider the polynomial equations $A y = 0$ with $y \in \CC^m$. Localizing at
$\det N \neq 0$, that is looking at these equations in the local ring
$\QQ[x,y]_{\det N}$, similarly to the proof of \cite[Proposition 3.2.7]{SaSc13},
one has that $A y = 0$ is equivalent to
$$
0' = y' \left[ \begin{array}{cc} N'  & P \\ Q' & R' \\ \end{array} \right] = 
y'
\left[ \begin{array}{cc} N'  & P \\ Q' & R' \\ \end{array} \right]
\left[ \begin{array}{cc} N'^{-1} & 0 \\ 0 & Id_{m-r} \\ \end{array} \right]
\left[ \begin{array}{cc} Id_{r}  & -P \\ 0 & Id_{m-r} \\ \end{array} \right]
$$
which equals
$$
0' = y'
\left[ \begin{array}{cc} Id_r  & 0 \\ Q'N'^{-1} & R'-Q'N'^{-1}P \\ \end{array} \right]
\qquad
\text{that is}
\qquad
0 = \left[ \begin{array}{cc} Id_r  & N^{-1}Q \\ 0 & Sch(N) \\ \end{array} \right]y
$$
where $Sch(N) = R-P'N^{-1}Q$ is the Schur complement of $N$ in $A$. Denoting
with $q_{i,j}$ the entries of $-N^{-1}Q$ and with $q'_{i,j}$ the entries of
$-Sch(N)$ ends the proof.
\end{proof}

We let $\widetilde{A}$ be the matrix 
$$\widetilde{A}=\left[
\begin{array}{cc}
  Id_{r}   & N^{-1}Q \\
  0_{m-r,r} & Sch(N) \\
\end{array}
\right]
$$
and 
$
p'=
\left[ p_1 \, \ldots \, p_m \right]' = 
\widetilde{A}
\left[ y_1 \, \ldots \, y_m \right]'
$
the new equations given by Lemma \ref{lemma:local}. One deduces
\begin{align*}
\left[ p_{1} \, \ldots \, p_r \right]'
& = \left[ y_{1} \, \ldots \, y_r \right]' +
(N^{-1}Q) \cdot \left[ y_{r+1} \, \ldots \, y_m \right]' \\
\left[ p_{r+1} \, \ldots \, p_m \right]'
& = Sch(N) \cdot \left[ y_{r+1} \, \ldots \, y_m \right]'.
\end{align*}
We deduce that modulo $\left \langle p_1, \ldots, p_r \right \rangle$ one can
express $y_1, \ldots, y_r$ as functions of $x$ and $y_{r+1}, \ldots, y_m$. By this, the
system defining $\setV$ in the open set defined by $\det N\neq 0$ can be re-written
in the local ring $\QQ[x,y]_{\det N}$ \begin{align*}
Sch(N) \cdot \left[ y_{r+1} \, \ldots \, y_m \right]' &= 0 \\
q(u, y_{r+1}, \ldots, y_m) &= 0
\end{align*}
where $q$ is a linear form on $y_i$, $i=r+1, \ldots, m$, parametrized by $u$,
and whose coefficients belong to $\QQ[x,y]_{\det N}$, obtained by substituting
in $u'y-1=0$ the expressions for $y_1, \ldots, y_r$. Since $q$ is linear w.r.t.
$y_{r+1}, \ldots, y_m$, one deduces that one can eliminate another variable
among $y_{r+1}, \ldots, y_m$, say w.l.o.g. $y_{r+1}=\ell(x,y_{r+2} \, \ldots \, y_m)$,
in the localization of $\QQ[x,y]_{\det N}$ to $e(x) \neq 0$, where $e(x) \in
\QQ[x,y]_{\det N}$ is the coefficient of $y_{r+1}$ in the polynomial $q$.
Finally, the system defining $\setV$, in the local ring $\QQ[x,y]_{(\det N) \cdot e}$
can be re-written as
$$
Sch(N) \cdot \left[ \ell(y_{r+2} \, \ldots \, y_m), y_{r+2}, \ldots, y_m \right]' = 0.
$$
We call $F$ this new system. Consider the polynomial system
\begin{equation} \label{eq:local-lag}
\left[ z_1 \, \ldots \, z_{m-r} \right] D_xF = w'
\end{equation}
with $w = (w_1, \ldots, w_n) \in \CC^n$.

\begin{lemma} \label{lemma:local-lag}
Suppose that $r \leq m-2$.
There exists a non-empty Zariski open set $\mathscr{W} \subset \CC^n$ such that if
$w \in \mathscr{W}$, the system \eqref{eq:local-lag} defines the empty set.
\end{lemma}
\begin{proof}
  Let $C \subset \CC^{n+2(m-r)-1}$ be the Zariski closure of the
  constructible set defined by \eqref{eq:local-lag}, $\det N(x) \neq
  0$ and $\rank \, A(x) = r$. Consider the projections $\pi_1 \colon C
  \to \CC^n$ defined by $\pi_1(x,y,z,w)=x$ and $\pi_2 \colon C \to
  \CC^n$ defined by $\pi_2(x,y,z,w)=w$. The image $\pi_1(C)$ is dense
  in the set $\{x \in \CC^n : \rank \, A(x) \leq r\}$ and so has
  dimension $\leq n-(m-r)^2$ when $A$ is generic. The fiber of $\pi_1$
  over a generic point $x \in \pi_1(C)$ is the graph of the functions
  $w_1, \ldots, w_n$ of $(y,z)$, and so it has co-dimension $n$ and
  dimension $[n+(m-r)+m-r-1]-n = 2(m-r)-1$. By the Theorem of the
  Dimension of Fibers \cite[Sect. 6.3, Theorem 7]{Shafarevich77} one
  concludes that $C$ has dimension less than or equal to
  $n-(m-r)^2+2(m-r)-1 = n-(m-r-1)^2$ and since $r \leq m-2$ this
  dimension is $\leq n-1$. Therefore $\pi_2(C) \subset \CC^n$ is a set
  whose Zariski closure has dimension $\leq n-1$ and is contained in a
  hypersurface $H \subset \CC^n$. Setting $\mathscr{W} = \CC^n
  \setminus H$ ends the proof.
\end{proof}

We recall in what follows the global definition of Lagrange system for the
restriction of the projection $\pi_w \colon (x,y) \to w'x$ to the incidence
variety $\setV(A, u) = \setV(A \circ Id_n, u)$. The set $\setV(A, u)$ is defined by
$f(A, u)$ which consists in the entries of $A(x)y$ and $u'y-u_{m+1}$
with $u_{m+1} \neq 0$. We suppose that $A \in \zarA$. The associated
Lagrange system consists in the polynomial entries of
\begin{equation} \label{lagrange-system}
f(A, u),
\qquad
(g,h)=
\left[ z_1 \, \ldots \, z_{m+1}, z_{\bullet} \right]
\left[
\begin{array}{cc}
D_xf & D_yf \\
w_1 \, \ldots \, w_n & 0
\end{array}
\right],
\qquad
\sum_{i=1}^{m+1}v_iz_i-1
\end{equation}
where $g=(g_1, \ldots, g_n)$ and $h=(h_1, \ldots, h_m)$. We denote by
$\setW_w(A,u,v)$ the zero set of \eqref{lagrange-system}.
Over a solution of \eqref{lagrange-system}, $z_{\bullet} \neq 0$, since $A \in
\zarA$ and $f$ satisfies Assumption $\sfG_2$. The polynomial system \eqref{lagrange-system}
consists of $n+2m+2$ polynomials in $n+2m+2$ variables.

\begin{lemma} \label{lemma:lagr:proj}
Let $A \in \zarA$ (see \ref{remark:sing}) and $u \neq 0$. There exist non-empty Zariski open sets $\mathscr{W} \subset
\CC^n$ and $\zarV \subset \CC^{m+1}$ such that, if $w \in \mathscr{W}$ and $v \in \zarV$
the following holds:
\begin{enumerate}
\item the Jacobian matrix of \eqref{lagrange-system} has maximal rank
  at any point of $\setW_w(A,u,v)$ and $\setW_w(A,u,v)$ is finite;
\item the projection of $\setW_w(A,u,v)$ in the space of $x,y$ contains the critical
  points of $\pi_w \colon (x,y) \to w'x$ restricted to $\setV(A,u)$.
\end{enumerate}
\end{lemma}
\begin{proof}[of Assertion (1) of Lemma \ref{lemma:lagr:proj}]
  The Lagrange system \eqref{lagrange-system} has a local description,
  as shown above, when we look at its equations in the local ring
  $\QQ[x,y,z]_{\det N}$ where $N$ is a given square submatrix of
  $A(x)$. If $\mathscr{W}' \subset \CC^n$ is the non-empty Zariski
  open set given by Lemma \ref{lemma:local-lag}, when $w \in
  \mathscr{W}'$ then if $(x,y,z)$ is a solution of
  \eqref{lagrange-system}, one deduces that the rank of $A(x)$ is
  $m-1$. Without loss of generality, we can assume to work in the open
  set $\det N \neq 0$ where $N$ is the $(m-1) \times (m-1)$ upper-left
  submatrix of $A(x)$ and work in the local ring $\QQ[x,y,z]_{\det
    N}$.  \smallskip

  Hence, let $N$ be the $(m-1) \times (m-1)$ upper-left submatrix of
  $A(x)$. By Lemma \ref{lemma:local}, the local equations of
  $\setV_r$ in $\QQ[x,y]_{\det N}$ are of the form
\[
y_i - q_i(x) y_m =0, \, i=1, \ldots, m-1 \qquad Sch(N) y_m = 0 \qquad u'y-u_{m+1} = 0.
\]
for some $q_i \in \QQ[x,y]_{\det N}$.

We prove by contradiction that $y_m\neq 0$; hence assume that
$y_m=0$. By the first $m-1$ equations one deduces that $y=0$, which is
a contradiction with $u'y-u_{m+1} = 0$ and $u_{m+1}\neq 0$. Then one
can suppose $y_m=1$. Moreover, since $N$ has size $m-1$, its Schur
complement is exactly $Sch(N) = \det(A) \cdot \det(N^{-1}) =
\det(A)/\det(N)$, and so in the local ring $\QQ[x,y]_{\det N}$ the
equation $Sch(N) y_m = 0$ is equivalent to $\det(A) = 0$. So one
obtains that the incidence variety $\setV(A,u)$ has the following
local description:
\[
\det(A) = 0 \qquad y_i - q_i(x) =0, \, i=1, \ldots, m-1 \qquad y_m - 1 = 0.
\]
We call $f = (f_1, \ldots, f_{m+1})$ these local equations. The Jacobian of $f$ reads
\[
Df = [D_xf \,\,\, D_yf] =
\left[
\begin{array}{cc}
D_x \det(A) & 0_{1 \times m} \\
\star & Id_{m}
\end{array}
\right]
\]
and the local equations of the associated Lagrange system are
$$
f=0,
\qquad
(g,h) =
[z_1 \,\, \ldots \,\, z_{m+1}, z_{\bullet}]
\left[
\begin{array}{cc}
D_xf & D_yf \\
w & 0
\end{array}
\right],
\qquad
\sum_{i=1}^{m+1}v_1z_i-1=0,
$$
which is also a square system with $n+2m+2$ polynomials and variables.
Now, consider the map
\[
  \begin{array}{lrcc}
  p : &  \CC^{n+2m+2} \times \CC^{n} \times \CC^{m+1} & \longrightarrow & \CC^{n+2m+2} \\
            &  (x,y,z,w,v) & \longmapsto & (f, g, h, \sum v_iz_i-1).
  \end{array}
\]
and its section map
\[
  \begin{array}{lrcc}
  p_{v,w} : &  \CC^{n+2m+2} & \longrightarrow & \CC^{n+2m+2} \\
            &  (x,y,z) & \longmapsto & (f, g, h, \sum v_iz_i-1).
  \end{array}
\]
for given $v \in \CC^{m+1}, w \in \CC^n$. Remark that $p_{v,w}^{-1}(0)=\setW_w(A,u,v)$
for all $v,w$. Now, there are two cases: either $\zeroset p = \emptyset$ or $\zeroset
p \neq \emptyset$. \smallskip

In the first case, this means in particular that for all $v,w$, $\setW_w(A,u,v) =
\emptyset$, which implies Assertion (1) with $\mathscr{W}=\CC^n$ and $\zarV=\CC^{m+1}$. \smallskip

Suppose now that $\zeroset p \neq \emptyset$.
We claim that $0$ is a regular value for $p$. Indeed, by Thom's Algebraic Weak
Transversality theorem \cite[Sect. 4.2]{SaSc13} there exist non-empty Zariski open
sets $\mathscr{W} \subset \CC^n$ and $\zarV' \subset \CC^{m+1}$ such that if
$\v \in \zarV'$ and $w \in \mathscr{W}'$, zero is a regular value of the map $p_{v,w}$.
The regularity property
of $p_{v,w}$ at $0$ implies that the rank of \eqref{lagrange-system} is maximal
(and equal to $n+2m+2$) at each point of $\setW_w(A,u,v)$ and, by the Jacobian
Criterion \cite[Theorem 16.19]{Eisenbud95}, that $\setW_w(A,u,v)$ is empty or
finite, implying that Assertion (1) is true. \smallskip

We prove that $0$ is a regular value for $p$. Let $(x,y,z,w,v) \in p^{-1}(0)$ and let $Dp$ be the jacobian
matrix of the map $p$ at $(x,y,z,w,v)$. Remark that:
\begin{itemize}
\item since $\rank \, A(x) = m-1$ and $A \in \zarA$, then $x$ is a
  regular point of $\det(A) = 0$;
\item since $A \in \zarA$ then $z_{\bullet} \neq 0$;
\item since $z_1, \ldots, z_{m+1}$ verifies $\sum v_iz_i=1$ then there exists
  $\ell \in \{1, \ldots, m+1\}$ such that $z_\ell \neq 0$;
\item polynomial entries of $h$ have the form $(h_1, \ldots, h_m) = (z_2, \ldots, z_{m+1})$.
\end{itemize}
We consider the submatrix of $Dp$ made by the following independent square blocks:
\begin{itemize}
\item the derivative $\frac{\partial}{\partial x_j}\det A$ which is non-zero at $x$;
\item the derivatives of $y_i-q_i(x), i=1, \ldots, m-1$ and $y_m-1$ with respect
  to $y_1, \ldots, y_m$;
\item the derivatives of $g$ w.r.t. $w_1, \ldots, w_n$ give a non-singular diagonal
  matrix $z_{\bullet}Id_{n}$;
\item the derivatives of $h$ w.r.t. $z_2, \ldots, z_{m+1}$ give a non-singular diagonal
  matrix $Id_{m}$;
\item the derivative of $\sum v_iz_i-1$ with respect to $v_\ell$ with $\ell$ as above
  is $z_{\ell} \neq 0$.
\end{itemize}

So one obtains that $Dp$ has a non-singular maximal minor of size $n+2m+2$ and
so it has full rank. By the genericity of the point $(x,y,z,w,v) \in p^{-1}(0)$
one deduces the claimed property.
\end{proof}

\begin{proof}[of Assertion (2) of Lemma \ref{lemma:lagr:proj}]
Suppose first that $\setW_w(A,u,v)=\emptyset$ for all $w \in \CC^n, v \in \CC^{m+1}$.
Fix $w \in \CC^n$, $(x,y) \in \setV(A,u)$ and suppose that $(x,y)$ is a critical point
of $\pi_w$ restricted to $\setV(A,u)$. Then, since $\setV(A,u)$ is equidimensional,
there exists $z \neq 0$ such that $(x,y,z)$ verifies the equations
$z'Df = [w,0]$. Since $z \neq 0$, there exists $v \in \CC^{m+1}$ such that
$v'z=1$. So we conclude that $(x,y,[z,1]) \in \setW_w(A,u,v)$, which is a contradiction. \smallskip

Suppose now that $A \in \zarA, u \neq 0$, $\zeroset p$ is non-empty
and that $w \in \mathscr{W}$. By \cite[Sect. 3.2]{SaSc13}, the set of
critical points of $\pi_w$ restricted to $\setV(A,u)$ is the
projection $\pi_{x,y}$ on $x,y$ of the constructible set:
$$
\mathcal{S} = \{(x,y,z) : f=g=h=0, \, z \neq 0\} 
$$
where $f,g,h$ have been defined in \eqref{lagrange-system}. One can
easily prove with the same techniques as in the proof of Assertion (1)
that $\mathcal{S}$ has dimension 1.  Moreover, for each $(x,y)$ in
$\pi_{x,y}(\mathcal{S})$, the fiber $\pi_{x,y}^{-1}(x,y)$ has
dimension 1 (because of the homogeneity of the $z$-variables). By the
Theorem on the Dimension of Fibers \cite[Sect. 6.3, Theorem
7]{Shafarevich77}, we deduce that $\pi_{x,y}(\mathcal{S})$ is
finite. Fix now $(x,y) \in \pi_{x,y}(\mathcal{S})$ and let
$\zarV_{(x,y)} \subset \CC^{m+1}$ be the non-empty Zariski open set of
$v$ such that the hyperplane $v'z-1=0$ intersects transversely
$\pi_{x,y}^{-1}(x,y)$, and recall that $\zarV' \subset \CC^{m+1}$ has
been defined in the proof of Assertion (1).  By defining
$$
\zarV = \zarV' \bigcap_{(x,y) \in \pi_{x,y}(\mathcal{S})}\zarV_{(x,y)}
$$
one concludes the proof.
\end{proof}

\subsection{Proof of Proposition \ref{prop:dimlag}} \label{subsec:proof-prop}

We prove now Proposition \ref{prop:dimlag}.

\begin{proof}[of Proposition \ref{prop:dimlag}]
Let $\zarM_1 \subset \GL(n,\CC)$ be the set of invertible matrices $\M$
such that the first row $w'$ of $\M^{-1}$ lies in the set $\mathscr{W}$
given in Lemma \ref{lemma:lagr:proj}. This is a non-empty Zariski open set of
$\GL(n,\CC)$ since the entries of $\M^{-1}$ are rational functions
of the entries of $\M$. Let $\zarV \subset \CC$
be the non-empty Zariski open set given by Lemma \ref{lemma:lagr:proj}
and let $v \in \zarV$.
Let $e'_1=(1, 0, \ldots, 0)\in \QQ^n$ and for all $\M \in \GL(n,\CC)$, let
\[ 
\tilde{\M} =\left(\begin{array}{cc}
  \M & {0} \\
{0} & Id_m \\
\end{array}\right). 
\]
Remark that for any $\M\in \zarM_1$ the following identity holds:
$$
\left(\begin{array}{c}
  \jac\f(\A \circ \M, u) \\
  e'_1 \quad 0\; \cdots \; 0\\  
\end{array}\right) = \left(\begin{array}{cc}
  \jac \f(\A,u) \\
w' \quad 0 \; \cdots \; 0\\
\end{array}\right)\tilde{\M}.
$$
We conclude that the set of solutions of the system 
\begin{equation}
  \label{eq:dim:1}
  \left(\f(\A, u), \quad
  z'\left(\begin{array}{cc}
  \jac_x \f & \jac_y \f \\
w' & 0 \; \cdots \; 0\\
\end{array}\right), \quad \sum_{i=1}^{m+1} v_iz_{i} -1 \right)
\end{equation}
is the image by the map $(\x, \y) \mapsto \tilde{\M}^{-1}(\x, \y)$ of the set $\mathcal{S}$ of solutions
of the system
\begin{equation}
  \label{eq:dim:2}
  \left(\f(\A \circ \M, u), \quad
 z'\left(\begin{array}{c}
      \jac \f(\A \circ \M, u) \\
      e'_1 \quad 0\; \cdots \; 0\\  
\end{array}\right), \quad \sum_{i=1}^{m+1} v_iz_{i} -1 \right).
\end{equation}
Now, let $\pi$ be the projection that forgets the last coordinate of $z$,
that is $\Z_{\bullet}$. Remark that $\pi(\mathcal{S})={\mathcal Z}(A\circ M,u,v)$
and that $\pi$ is a bijection. Moreover, it is an isomorphism of affine algebraic
varieties, since if $(x,y,z) \in {\cal S}$, then its $\Z_{\bullet}$-coordinate is
obtained by evaluating a polynomial at $(x,y,z_1 \ldots z_{m+1})$.

Thus, Assertion (1) of Lemma \ref{lemma:lagr:proj} implies that $\mathcal{S}$
and $\pi(\mathcal{S})={\mathcal Z}(A\circ M,u,v)$ are finite which proves
Assertion (1) of Proposition \ref{prop:dimlag}.

Assertion (1) of Lemma \ref{lemma:lagr:proj} also implies that the
Jacobian matrix associated to \eqref{eq:dim:2} has maximal rank at any
point of $\mathcal{S}$. Since we already observed that
$\pi(\mathcal{S})={\cal Z}(\A \circ \M, u, v)$ and that the
map is an isomorphism, Assertion (2) follows.

Assertion (3) is a straightforward consequence of Assertion (2) of
Lemma \ref{lemma:lagr:proj}.
\end{proof}

\section{Closure properties of projection maps}\label{sec:closedness}

The goal of this section is to prove Proposition
\ref{prop:closedness}.
We start by introducing some notations.

\begin{notations} \label{not:one}

For an algebraic set ${\mathcal Z} \subset \CC^n$ of dimension $d$, we denote
by $\Omega_i(\mathcal Z)$ the $i$-equidimensional component of $\mathcal Z$, for $i=0,1, \ldots, d$.

We denote by $\scS(\mathcal Z)$ the union of the following sets:
\begin{itemize}
\item $\Omega_0(\mathcal Z) \cup \cdots \cup \Omega_{d-1}(\mathcal Z)$
\item the set $\sing(\Omega_d(\mathcal Z))$ of singular points of $\Omega_d(\mathcal Z)$
\end{itemize}
and by $\scC(\pi_i, \mathcal Z)$ the Zariski closure of the union of the
following sets:
\begin{itemize}
\item $\Omega_0(\mathcal Z) \cup \cdots \cup \Omega_{i-1}(\mathcal Z)$;
\item the union for $r \geq i$ of the sets $\crit(\pi_i, \reg(\Omega_r(\mathcal Z)))$ of
  critical points of the restriction of $\pi_i$ to the regular locus of $\Omega_r(\mathcal Z)$.
\end{itemize}
Now, take $\M \in \GL(n,\CC)$ and fix ${\mathcal Z} \subset \CC^n$ algebraic
set of dimension $d$.
We define the collection of algebraic sets $\{{\mathcal O}_i(\M^{-1}{\mathcal Z})\}_{0 \leq i \leq d}$ with
\begin{itemize}
\item ${\mathcal O}_d(\M^{-1}{\mathcal Z})=\M^{-1}{\mathcal Z}$;
\item ${\mathcal O}_i(\M^{-1}{\mathcal Z})=\scS({\mathcal O}_{i+1}(\M^{-1}{\mathcal Z}))
\cup \scC(\pi_{i+1},  {\mathcal O}_{i+1}(\M^{-1}{\mathcal Z})) \cup
\scC(\pi_{i+1},\M^{-1}{\mathcal Z})$ for $i=0, \ldots, d-1$.
\end{itemize}
\end{notations}

\paragraph*{Property $\sfP(\mathcal Z)$.} Let ${\mathcal Z} \subset \CC^n$ be an algebraic set
of dimension $d$.
We say that $\M \in \GL(n,\CC)$ satisfies $\sfP(\mathcal Z)$ when for all
$i=0, 1, \ldots, d$
\begin{enumerate}
\item ${\mathcal O}_i(\M^{-1}{\mathcal Z})$ has dimension $\leq i$; 
\item ${\mathcal O}_i(\M^{-1}{\mathcal Z})$ is in Noether position with respect to $X_1, \ldots, X_i$.
\end{enumerate}

Note that Point (2) of $\sfP(\mathcal Z)$ implies Point (1) (this is an
immediate consequence of \cite[Chap. 1.5.3]{Shafarevich77}).
The following result shows that Property
$\sfP(\mathcal Z)$ holds for a generic choice of the matrix $\M$ and it will be
proved later on.

\begin{proposition} \label{prop:propertyP}
  Let ${\mathcal Z} \subset \CC^n$ be an algebraic set of dimension $d$.
  There exists a non-empty Zariski open set $\zarM_2 \subset \GL(n,\CC)$ such
  that for all $\M \in \zarM_2$, $\M$ satisfies $\sfP(\mathcal Z)$. 
\end{proposition}

\paragraph*{Property $\sfQ(\mathcal Z)$.} Let $\mathcal Z$ be an algebraic set of
dimension $d$ and $1\leq i \leq d$. We say that $\sfQ_i(\mathcal Z)$ holds if
for any connected component $\mathcal C$ of ${\mathcal Z}\cap \RR^n$ the boundary of
$\pi_i(\mathcal C)$ is contained in $\pi_i({\mathcal O}_{i-1}(\mathcal Z) \cap {\mathcal C})$. When there
is no ambiguity on $\mathcal Z$, we simply write that $\sfQ_i$ holds.

The following result describes properties of projections of the
connected components of the real counterpart of an algebraic set when 
property $\sfP(\mathcal Z)$ holds. 

\begin{proposition} \label{prop:frontier} Let ${\mathcal Z} \subset \CC^n$ be an
  algebraic set of dimension $d$ and $\M \in \GL(n,\CC) \cap \QQ^{n
    \times n}$.  If $\M$ satisfies $\sfP(\mathcal Z)$, then 
  $\sfQ_1(\M^{-1}\mathcal Z), \ldots, \sfQ_d(\M^{-1}\mathcal Z)$ hold.
\end{proposition}

The relationship of Noether position with closedness properties of
connected components of real counterparts in algebraic sets and
critical points is already exhibited and exploited in
\cite{SaSc03}. Actually, Propositions \ref{prop:propertyP} and
\ref{prop:frontier} are already proved in \cite{SaSc03} under the
assumption that $\mathcal Z$ is smooth and equidimensional. We cannot make this
assumption in our context to prove Proposition \ref{prop:closedness}
since $\setD$ is generically singular. Thus, this Section can be seen
as a strict generalization of \cite{SaSc03}.

As in \cite{SaSc03}, we use the notion of proper map. A map
$p : {\mathcal U} \subset \CC^n\to \CC^i$ is {\it proper at} $y\in \CC^i$
if and only if there exists a neighbourhood $\mathcal B$ of $y$ such that
$p^{-1}(\overline{\mathcal B})\cap {\mathcal U}$ is closed and bounded where
$\overline{\mathcal B}$ is the closure of $\mathcal B$ for the strong topology. We
simply say that $p$ is proper when it is proper at any point of
$\CC^i$.

  \begin{proof}[of Proposition \ref{prop:frontier}]
    To keep notations simple, we suppose that $Id_n$ satisfies
    $\sfP(\mathcal Z)$. Our reasoning is by decreasing induction on the index
    $i$. In the whole proof we also define the following function on
    $\mathcal Z$: we associate to $\y \in \mathcal Z$ the value
    $$
    J(\y) = \min \big\{ j \mid \y \in {\mathcal O}_j \big\}.
    $$

    We start by establishing that $\sfQ_d$ holds. Let $\x \in \RR^d$
    be on the boundary of $\pi_d(\mathcal C)$. By \cite[Lemma 3.10]{Jelonek},
    Property $\sfP(\mathcal Z)$ implies that the map $\pi_d$ restricted to
    ${\mathcal O}_d(\mathcal Z)$ is proper, and so closed.  We deduce that the
    restriction of $\pi_d$ to ${\mathcal O}_d(\mathcal Z) \cap {\mathcal C} \cap {\mathcal Z} = {\mathcal O}_d({\mathcal Z}) \cap {\mathcal C}$
    is closed and that $\x \in \pi_d({\mathcal O}_d(\mathcal Z) \cap {\mathcal C})$. Let $\y \in
    {\mathcal O}_d({\mathcal Z}) \cap {\mathcal C}$ such that $\pi_d(\y)=\x$. If $J(\y) \leq d-1$ our
    conclusion follows immediately. Suppose now that $J(\y)=d$. This
    implies that $\y \in \reg \Omega_d(\mathcal Z)$. By the Implicit  Function
    Theorem we conclude that $\y$ is a critical point of $\pi_d$ and
    that $\y \in \crit(\pi_d, \reg(\Omega_d(\mathcal Z))) \subset \scC(\pi_d,
    {\mathcal Z}) \subset {\mathcal O}_{d-1}(\mathcal Z)$, which is a contradiction since we assumed $J(y)=d$.

    Suppose now that $\sfQ_{i+1}$ holds. We proceed in two steps:

    \begin{enumerate}

    \item First, we prove that the boundary of $\pi_i(\mathcal C)$ is included in 
      $\pi_i({\mathcal O}_i(\mathcal Z) \cap {\mathcal C})$. Indeed, let $\x \in \RR^i$ be on the boundary
      of $\pi_i(\mathcal C)$. Let $p \colon \RR^{i+1} \to \RR^i$ be the map sending
      $(\x_1, \ldots \x_{i+1})$ to $(\x_1, \ldots \x_{i})$, so that $\pi_i = 
      p \circ \pi_{i+1}$. For $r>0$, let ${\mathcal B}_r$ be the ball of
      center $\x$ and radius $r$ in $\RR^i$ and ${\mathcal B}'_r = p^{-1}({\mathcal B}_r)$.
      We claim that ${\mathcal B}'_r$ meets both $\pi_{i+1}(\mathcal C)$ and its complementary
      in $\RR^{i+1}$. 

Indeed this is a consequence of the following immediate equalities 
\[
\pi_i^{-1}({\mathcal B}_r) \cap {\mathcal C} = 
      \pi_{i+1}^{-1} \circ p^{-1}({\mathcal B}_r) \cap {\mathcal C} = \pi_{i+1}^{-1}({\mathcal B}'_r) 
      \cap {\mathcal C}
\]
and $\pi_i^{-1}({\mathcal B}_r) \cap {\mathcal C}\neq \emptyset $ and ${\mathcal B}_r \cap \{\RR^i
\setminus \pi_i({\mathcal C})\} \neq \emptyset$.  Since ${\mathcal B}'_r$ is connected,
${\mathcal B}'_r$ meets also the boundary of $\pi_{i+1}(\mathcal C)$. Since ${\sfQ}_{i+1}$
holds, for every $r>0$ there exists $\y_r \in {\mathcal O}_i(\mathcal Z) \cap {\mathcal C}$ such that
$\pi_{i+1}(\y_r) \in {\mathcal B}'_r$, and so $\pi_i(\y_r) \in {\mathcal B}_r$. Thus, $\x$
lies in the closure of the image by $\pi_i$ of the set ${\mathcal O}_i(\mathcal Z) \cap
{\mathcal C}$. This image is closed and our claim follows.

\item Second, we prove that $\sfQ_i$ holds. Let $\x \in \RR^i$ be on the boundary
  of $\pi_i(\mathcal C)$. From (1), we deduce that there exists $\y
  \in {\mathcal O}_i(\mathcal Z) \cap {\mathcal C}$ such that $\pi_i(\y)=\x$.  Suppose by
  contradiction that for all $y$ as above, $J(\y)=i$. Fix $\y \in
  {\mathcal O}_i(\mathcal Z) \setminus {\mathcal O}_{i-1}(\mathcal Z)$ such that $\pi_i(\y)=\x$.  In
  particular, $\y \in {\mathcal O}_i(\mathcal Z) \setminus \scS({\mathcal O}_i(\mathcal Z))$, and thus, we
  deduce that $\y \in \reg(\Omega_i({\mathcal O}_i))$. Next, since $\x \in
  \pi_i(\Omega_i({\mathcal O}_i) \cap {\mathcal C})$ and lies on the boundary of $\pi_i(\mathcal C)$,
  we deduce that $\x$ lies on the boundary of $\pi_i(\Omega_i({\mathcal O}_i) \cap
  {\mathcal C})$. Finally, by the Implicit Function Theorem, we deduce that $\y
  \in \crit(\pi_i, \reg {\mathcal O}_i) \subset \scC(\pi_i,  {\mathcal O}_i) \subset
  {\mathcal O}_{i-1}$, which is a contradiction since we assumed that $J(\y)=i$.
    \end{enumerate}
    We conclude that $\sfQ_i$ holds, and so all statements 
    $\sfQ_1, \ldots, \sfQ_d$ hold.
  \end{proof}

  \begin{lemma} \label{lemma:fiber} Let ${\mathcal Z} \subset \CC^n$ be an
    algebraic set.  Let $\M \in \GL(n,\CC)$ be such that $\M$
    satisfies $\sfP(\mathcal Z)$. Let $\M^{-1}\mathcal C$ be a connected component of
    $\M^{-1}{\mathcal Z} \cap \RR^n$ and $\fiber \in \RR$ be on the boundary of
    $\pi_1(\M^{-1}{\mathcal C})$. Then $\pi_1^{-1}(\fiber) \cap \M^{-1}{\mathcal C}$ is a
    non-empty finite set contained in ${\mathcal O}_0(\M^{-1}{\mathcal Z}) \cap \M^{-1}{\mathcal C}$.
  \end{lemma}

\begin{proof}[of Lemma \ref{lemma:fiber}]
  By Proposition \ref{prop:frontier} we deduce that if $\fiber \in
  \RR$ belongs to the boundary of $\pi_1(\mathcal C)$, there exists $\x \in
  {\mathcal O}_0(\M^{-1}{\mathcal Z}) \cap \M^{-1}{\mathcal C}$ such that $\pi_1(\x)=\fiber$.  So
  $({\mathcal O}_0(\M^{-1}{\mathcal Z}) \cap \M^{-1}{\mathcal C}) \cap (\pi_1^{-1}(\fiber) \cap \M^{-1}{\mathcal C})
  \neq \emptyset$.  Now, we prove that $\pi_1^{-1}(\fiber) \cap \M^{-1}{\mathcal C}
  \subset {\mathcal O}_0(\M^{-1}{\mathcal Z}) \cap \M^{-1}{\mathcal C}$. Since $\M$ satisfies $\sfP(\mathcal Z)$,
  ${\mathcal O}_0(\M^{-1}{\mathcal Z})$ is finite and we also deduce that $\pi_1^{-1}(\fiber)
  \cap \M^{-1}{\mathcal C}$ is finite.

  We use again the definition of the function $\x \mapsto J(\x)$ over
  $\mathcal Z$ used in the proof of Proposition \ref{prop:frontier}. Suppose
  that there exists $\x \in \pi_1^{-1} (\fiber) \cap \M^{-1}{\mathcal C}$ such that
  $J(\x)=j>0$; this implies that $\x \in {\mathcal O}_j(\mathcal Z) \setminus {\mathcal O}_{j-1}(\mathcal Z)$.
  In particular, we deduce that $\x \in \reg(\Omega_j({\mathcal O}_j(\mathcal Z))) \cap
  \M^{-1}{\mathcal C}$. Since $\fiber=\pi_1(\x)$ is on the boundary of
  $\pi_1(\M^{-1}{\mathcal C})$, we conclude that $\pi_j(\x)$ is on the boundary of
  $\pi_j(\Omega_j({\mathcal O}_j(\M^{-1}{\mathcal Z})) \cap \M^{-1}{\mathcal C})$. Moreover, since $\x \in
  \reg \Omega_j({\mathcal O}_j(\M^{-1}{\mathcal Z})) \cap \M^{-1}{\mathcal C}$, we conclude by the Implicit
  Function Theorem that $\x$ is a critical point of the restriction of
  $\pi_j$ to ${\mathcal O}_j(\M^{-1}{\mathcal Z})$. So $\x \in \crit(\pi_j, {\mathcal O}_j(\M^{-1}{\mathcal Z}))
  \subset \scC(\pi_j, {\mathcal O}_j(\M^{-1}{\mathcal Z})) \subset {\mathcal O}_{j-1}(\M^{-1}{\mathcal Z})$.  We
  conclude that contradiction $J(\x) \leq j-1$ which is a
  contradiction.
\end{proof}

We are now able to prove Proposition \ref{prop:closedness}.

\begin{proof}[of Proposition \ref{prop:closedness}]
  Let $\zarM_2 \subset \GL(n,\CC)$ be the non-empty Zariski open set
  of matrices satisfying Property $\sfP(\mathcal Z)$ defined in Proposition
  \ref{prop:propertyP}, and let $\M \in \zarM_2$. Let $\M^{-1}{\mathcal C}$ be a
  connected component of $\M^{-1}\setD \cap \RR^n$, and let $1 \leq i
  \leq n-1$. Then, applying Proposition \ref{prop:frontier}, we
  conclude that $\sfQ_i(\M^{-1}\setD)$ holds. In particular the boundary
  of $\pi_i(\M^{-1}{\mathcal C})$ is contained in $\pi_i({\cal O}_{i-1}(\setD^\M) \cap
  \M^{-1}{\mathcal C})\subset \pi_i(\M^{-1}{\mathcal C})$ which implies that $\pi_i(\M^{-1}{\mathcal C})$ is
  closed. This proves Assertion (1).  

  We prove now Assertion (2). Take $\fiber \in \RR$ that lies in the
  frontier of $\pi_1(\M^{-1}{\mathcal C})$. By Lemma \ref{lemma:fiber},
  $\pi_1^{-1}(\fiber) \cap \M^{-1}{\mathcal C}$ is a finite set, and thus there
  exists $\x \in \M^{-1}\setD \cap \RR^n$ such that $\x \in \M^{-1}{\mathcal C}$ and
  $\pi_1(x)=\fiber$. For all such $\x$, the matrix $\A(\x)$ is rank
  defective. Fix $\x \in \pi_1^{-1}(\fiber) \cap \M^{-1}{\mathcal C}$ and let
  $r\leq m-1$ be the rank of $\A(\x)$. 
  Consider the linear system $y \mapsto f(\A,u)$ parametrized by the vector $u$.
  This system has at least one solution $y$ if and only if
  $$
  \rank
  \left [\begin{array}{c}
      \A(\X) \\
      u_1 \; \cdots \; u_m \\  
    \end{array}\right ] =
  \rank
  \left [\begin{array}{cc}
      \A(\X) & {0} \\
      u_1 \; \cdots \; u_m & u_{m+1} \\
    \end{array}\right ].
  $$
  Now, the second matrix has rank $r+1$ since $u_{m+1} \neq 0$, and the first matrix has rank $r+1$ if
  and only if $(u_1, \ldots, u_m)$ does not lie in the space generated by the rows of $A$. So
  there exists a non-empty Zariski open set $\zarU_{\mathcal{C},x}$ such that if $(u_1, \ldots, u_m) \in \zarU_{\mathcal{C},x}$
  the linear system has at least one solution.

  We conclude the proof by taking
  $$
  \zarU = \bigcap_{{\cal C} \subset \setD\cap\RR^n}\bigcap_{x \in \pi_1^{-1}(\fiber) \cap \M^{-1}{\mathcal C}} \zarU_{\mathcal{C},x}
  $$
  which is non-empty and Zariski open because of the finiteness of
  $\pi_1^{-1}(\fiber) \cap \M^{-1}{\mathcal C}$ and of the number of
  connected components of $\setD \cap \RR^n$.
\end{proof}

The remainder of this Section is
dedicated to the proof of Proposition \ref{prop:propertyP}.  We start
by introducing some notations.

\begin{notations} \label{not:two} Let $\fermU$ be an $n$-by-$n$
  matrix of indeterminates.  For $f \in \QQ[\X_1, \ldots,
  \X_n]$, let $f\circ\fermU\in\QQ(\fermU)[\X_1, \ldots, \X_n]$ denote the
  polynomial such that $(f\circ\fermU)(\X)=f(\fermU\X)$, and if ${\mathcal V} \subset
  \CC^n$ is defined by the ideal $I = \langle f_1, \ldots, f_s
  \rangle$, let $\fermU^{-1}{\mathcal V}$ be the algebraic set defined by $I\circ\fermU
  = \langle f_1\circ\fermU, \ldots, f_s\circ\fermU \rangle \subset \QQ(\fermU)
  [\X_1, \ldots, \X_n]$.

  For all $i=0, 1, \ldots, d$, we denote by $I_i$, $I_i\circ\M $ and
  $I_i\circ\fermU $ the ideals associated to the algebraic sets ${\mathcal O}_i(\mathcal Z),
  {\mathcal O}_i(\M^{-1}{\mathcal Z})$ and ${\mathcal O}_i(\fermU^{-1}{\mathcal Z})$,
 see Notations \ref{not:one}. 

\end{notations}

\begin{lemma} \label{lemma:noeth1} Let ${\mathcal Z} \subset \CC^n$ be an
  algebraic set of dimension $d$ and $0\leq i\leq d$. Let $\fermP$ be
  one of the components of the prime decomposition of $I\circ\fermU_i$ and
  let $r = \dim \fermP$.  Then $r \leq i$ and the ring extension
  $\QQ(\fermU)[\X_1 \dots \X_{r}] \longrightarrow
  \bigslant{\QQ(\fermU) [\X_1 \dots \X_{n}]}{\fermP}$ is integral.
\end{lemma}

This Lemma is a generalization of \cite[Prop.1]{SaSc03} to the
non-equidimensional case. Its proof shares similar techniques than
those used for proving \cite[Prop.1]{SaSc03}. It exploits the
properties of the geometric objects defined in Notations \ref{not:one}
to retrieve an equidimensional situation. We sketch below the main
differences and will refer to the proof of \cite[Prop. 1]{SaSc03} for
the steps that are identical.

\begin{proof}[of Lemma \ref{lemma:noeth1}]
Our reasoning is by decreasing induction on the index $i$.

Suppose first that $i=d$, so that $I_d\circ\fermU=\ideal{\fermU^{-1}{\mathcal Z}}$
(by definition ${\mathcal O}_d(\fermU^{-1}{\mathcal Z})=\fermU^{-1}{\mathcal Z})$.
 Let $\fermP$ be a prime ideal of the prime decomposition of $I_d\circ\fermU$, and let
$r=\dim\fermP$. Thus, the algebraic set defined by $\fermP$ is an
irreducible component of dimension $r \leq d$ and then $\zeroset{\fermP} \subset
\Omega_r(\fermU^{-1}{\mathcal Z})$. By the Noether normalization lemma
\cite{logar1989computational}, the statement follows.

Suppose now that the statement is true for $i+1$. To simplify
notations we write ${\mathcal O}_i$ instead of ${\mathcal O}_i(\fermU^{-1}{\mathcal Z})$. In particular, we
assume that ${\mathcal O}_{i+1}$ has dimension $\leq i+1$. Consider the ideal
$I_i\circ\fermU$; using the definitions of the geometric objects
introduced in Notations \ref{not:one} one obtains the following
equalities:
$$
I_i\circ\fermU = \ideal{\scS({\mathcal O}_{i+1})} \cap \ideal{\scC(\pi_{i+1},
{\mathcal O}_{i+1})} \cap \ideal{\scC(\pi_{i+1}, \fermU^{-1}{\mathcal Z})}.
$$
Now, let $\fermP$ be a prime ideal associated to $I_i\circ\fermU$. Then,
$\fermP$ is a prime ideal associated to one of the three ideals in the
above intersection. We investigate below the three possible cases: 
\begin{enumerate}
\item $ \ideal{\scS({\mathcal O}_{i+1})}\subset \fermP $. Let $r = \dim \fermP$.
  In this case, we obtain
  $$
  \fermP \supset \ideal{\Omega_0({\mathcal O}_{i+1})} \cap \cdots \cap
  \ideal{\Omega_i({\mathcal O}_{i+1})} \cap \ideal{\sing(\Omega_{i+1}({\mathcal O}_{i+1}))}.
  $$
 Combined with the fact that $\fermP$ is prime, this implies that 
  \begin{itemize}
  \item either $ \ideal{\Omega_j({\mathcal O}_{i+1})}\subset \fermP$, for some
    $0\leq j\leq i$ ; then one gets $r \leq i$ and by the induction
    assumption that the extension $\QQ(\fermU)[\X_1 \dots \X_{r}]
    \longrightarrow \bigslant{\QQ(\fermU) [\X_1 \dots
      \X_{n}]}{\fermP}$ is integral ;
  \item or $ \ideal{\sing(\Omega_{i+1}({\mathcal O}_{i+1}))}\subset \fermP $. 
  \end{itemize}
  Assume that $ \ideal{\sing(\Omega_{i+1}({\mathcal O}_{i+1}))}\subset \fermP $.
  We deduce that $$\dim (\fermP) \leq \dim
  (\sing(\Omega_{i+1}({\mathcal O}_{i+1}))).
  $$ Since $\dim(\Omega_{i+1}({\mathcal O}_{i+1}))=i+1$ by definition, it follows that $$\dim
  (\sing(\Omega_{i+1}({\mathcal O}_{i+1})))\leq i $$ and we deduce
  that $\dim(\fermP)\leq i$.  Let $\f\circ\fermU = (f_1\circ\fermU,
  \ldots, f_s\circ\fermU)$ be a set of generators of the ideal
  associated to $\Omega_{i+1}({\mathcal O}_{i+1})$. Then
  $$
  \ideal{\sing(\Omega_{i+1}({\mathcal O}_{i+1}))} = \sqrt{\langle \f\circ\fermU, g_1, \ldots, g_N \rangle}
  $$
  where $g_1, \ldots, g_N$ are the minors of size $(n-i-1) \times
  (n-i-1)$ of the Jacobian matrix $\jac \f \circ \fermU$. 
  We prove below by induction on $t$ that for any prime $\fermQ$
  associated to $\langle \f\circ\fermU, g_1, \ldots, g_t\rangle$,
  the extension
  $$\QQ(\fermU)[\X_1 \dots \X_{r}] \longrightarrow
  \bigslant{\QQ(\fermU) [\X_1 \dots \X_{n}]}{\fermQ}$$ is
  integral. Taking $t=N$ will conclude the proof.

  For $t=0$, the induction assumption implies that for any prime
  $\fermQ$ associated to $\langle f\circ\fermU\rangle$, the extension
  $\QQ(\fermU)[\X_1 \dots \X_{r}] \longrightarrow
  \bigslant{\QQ(\fermU) [\X_1 \dots \X_{n}]}{\fermQ}$ is integral.

  Assume now that for any prime $\fermQ'$ associated to 
  $\langle \f\circ\fermU, g_1, \ldots, g_t\rangle$,
  the extension
  $$\QQ(\fermU)[\X_1 \dots \X_{r}] \longrightarrow
  \bigslant{\QQ(\fermU) [\X_1 \dots \X_{n}]}{\fermQ'}$$ is integral.

  We prove below that for any prime $\fermQ$ associated to
  $\langle \f\circ\fermU, g_1, \ldots, g_{t+1}\rangle$, the
  extension
  $$\QQ(\fermU)[\X_1 \dots \X_{r}] \longrightarrow
  \bigslant{\QQ(\fermU) [\X_1 \dots \X_{n}]}{\fermQ}$$ is integral.

  Remark that any prime $\fermQ$ associated to $\langle
    \f\circ\fermU, g_1, \ldots, g_{t+1}\rangle$ is a prime associated to
  $\fermQ'+\langle g_{t+1}\rangle$. Suppose that $g_{t+1}\notin
  \fermQ'$ (otherwise, the conclusion follows immediately) and let
  $r'$ be the Krull dimension of $\fermQ'$.

  By Krull's Principal Ideal Theorem, $\fermQ' + \langle g_{t+1}
  \rangle$ is equidimensional of dimension $r'-1$. Following {\it
    mutatis mutandis} the same argumentation as in the proof of
  \cite[Prop. 1]{SaSc03}, the ideal $\fermQ' + \langle g_{t+1} \rangle$
  contains a monic polynomial in $\X_{r'}$, so that the extension
  $$
  \QQ(\fermU)[\X_1 \dots \X_{r'-1}] \longrightarrow
  \bigslant{\QQ(\fermU)[\X_1 \dots \X_{n}]}{\fermQ' + \langle g_{t+1}
    \rangle}
  $$
  is integral. Our claim follows. 

\item $ \ideal{\scC(\pi_{i+1}, {\mathcal O}_{i+1}(\fermU^{-1}{\mathcal Z}))}\subset \fermP $.

  Recall that $\scC(\pi_{i+1}, {\mathcal O}_{i+1}(\fermU^{-1}{\mathcal Z}))$ is the union of
  $\crit(\pi_{i+1}, \reg(\Omega_{i+1}({\mathcal O}_{i+1})))$ and of the sets
  $\Omega_{j}(\fermU^{-1}{\mathcal Z})$ for $0\leq j \leq i$. When
  $\ideal{\Omega_{j}(\fermU^{-1}{\mathcal Z})}\subset \fermP$, one can apply the
  induction assumption.

  Thus, we focus on the case where $\ideal{\crit(\pi_{i+1},
  \reg(\Omega_{i+1}({\mathcal O}_{i+1})))} \subset \fermP$.

  The ideal $\ideal{\crit(\pi_{i+1}, \reg(\Omega_{i+1}({\mathcal O}_{i+1})))}$ is
  built as follows. Suppose that $\f\circ\fermU = (f_1\circ\fermU, \ldots,
  f_s\circ\fermU)$ defines $\ideal{\Omega_{i+1}({\mathcal O}_{i+1})}$, that $g_1,
  \ldots, g_N$ are the square minors of size $n-i-1$ of the Jacobian matrix
  of $\f \circ \fermU$ where the first $i$ columns are eliminated,
  and that $J$ is the ideal
  $\ideal{\sing(\Omega_{i+1}({\mathcal O}_{i+1}))}$. The following equality is
  immediate:
  $$
  \ideal{\crit(\pi_{i+1}, \reg(\Omega_{i+1}({\mathcal O}_{i+1})))} =
  \sqrt{\f\circ\fermU + \langle g_1, \ldots, g_N \rangle} : J^\infty,
  $$
  where, if $K,L$ are two ideals in the same ring $R$, then
  $K:L^{\infty} = \{p \in R \mid L^Np \subset K, \ \exists N \in \mathbb{N}\}$.
  We deduce that the ideal $\fermP$ is a prime component of
  $\sqrt{\f\circ\fermU + \langle g_1, \ldots, g_N \rangle}$ whose zero
  locus is not included in $\sing(\Omega_{i+1}({\mathcal O}_{i+1}))$. The
  integral ring extension property is already proved (by induction)
  for every component of the ideal $\langle \f\circ\fermU \rangle$; so we
  proceed as in the first point.

\item $\ideal{\scC(\pi_{i+1}, \fermU^{-1}{\mathcal Z})}\subset \fermP $. 

  Again, recall that $\scC(\pi_{i+1}, \fermU^{-1}{\mathcal Z})$ is the union of
  $\Omega_{j}(\fermU^{-1}{\mathcal Z})$ for $0\leq j\leq i$ and the union for $r'
  \geq i$ of the sets $\crit(\pi_i, \reg(\Omega_{r'}(\mathcal Z)))$ of critical
  points of the restriction of $\pi_i$ to the regular locus of
  $\Omega_{r'}(\fermU^{1}{\mathcal Z})$.

  Let $r' \geq i+1$, and $\Omega_{r'}(\mathcal Z)$ be the equidimensional
  component of $\mathcal Z$ of dimension $r'$. So we can assume $
  \ideal{\crit(\pi_{i+1}, \reg(\Omega_{r'}(\fermU^{-1}{\mathcal Z})))}\subset \fermP
  $. The proof follows exactly the same argumentation as the one in
  the second point.
\end{enumerate}
\end{proof}

The following lemma plays the same role as the one in
\cite[Prop. 2]{SaSc03}. It shows that there exists the integral
extension property in Lemma \ref{lemma:noeth1} is maintained when
specializing $\fermU$ to a generic matrix $\M$ of $\GL(n,\CC)$.  The
proof of the lemma below is exactly the same as the one of
\cite[Prop. 2]{SaSc03}.

\begin{lemma} \label{lemma:noeth2} Let ${\mathcal Z} \subset \CC^n$ be an
  algebraic set of dimension $d$.  There exists a non-empty Zariski
  open set $\zarM_{2} \subset \GL(n,\CC)$ such that if $\M \in
  \zarM_{2} \cap \QQ^{n\times n}$, the following holds. Let $i \in
  \{0,1,\dots,d\}$ and $\fermP$ be a prime component of
  $I_i\circ{\M}$ and let $r = \dim (\fermP)$.  Then $r \leq i$ and
  the ring extension $\CC[\X_1 \dots \X_{r}] \longrightarrow
  \bigslant{\CC[\X_1 \dots \X_{n}]}{\fermP}$ is integral.
\end{lemma}

Now we can prove Proposition \ref{prop:propertyP}.

\begin{proof}[of Proposition \ref{prop:propertyP}] \label{propertyP}
  Let $\zarM_{2} \subset \GL(n,\CC)$ be the non-empty Zariski open set
  defined in Lemma \ref{lemma:noeth2}. By Lemma \ref{lemma:noeth2},
  for $\M \in \zarM_{2}$ and $0\leq i\leq d$, any irreducible
  component of the algebraic set $O_i(\M^{-1}Z)$ is in Noether position
  with respect to $\X_1, \ldots, \X_i$.  This proves Point (2) of
  $\sfP(Z)$. 
  Now, remark that \cite[Chap. 1.5.3]{Shafarevich77} implies that any
  irreducible component of $O_i(\M^{-1}Z)$ has dimension $\leq i$. This
  proves Point (1) of $\sfP(Z)$.
\end{proof}

\section{Practical experiments}\label{sec:experiments}

In this section, we report on practical experiments done with a computer
implementation of our algorithm. 

We have implemented the algorithm {\sf RealDet} under
\textsc{Maple}. The computation of rational parametrizations is done
using Gr\"obner bases, see \cite{F4, F5, FM11, FM13,newfglm}. We use
the Gr\"obner basis library \textsc{FGb} \cite{faugere2010fgb}
implemented in {\tt C} by J.-C. Faug\`ere and its interface with {\sc
  Maple}.

We mainly compare our implementation of {\sf RealDet} with the Real
Algebraic Geometry Library \textsc{RAGlib} \cite{raglib} implemented
by the last author. \textsc{RAGlib} is also a {\sc Maple} library
implementing algorithms based on the critical point method. It also
uses Gr\"obner bases and the library \textsc{FGb} for solving
polynomial systems of dimension $0$. We use its command {\tt
  PointsPerComponents} to compute sample points in each connected
component of the real counterpart of the hypersurface defined by the
vanishing of the determinant of the matrix under consideration. We
also made some tests using implementations of the Cylindrical
Algebraic Decomposition but none of them succeeded to solve the
examples we report on below.

The computations we report on have been performed on an 
Intel(R) Xeon(R) CPU $E7540@2.00{\rm GHz}$
256 Gb of RAM. The symbol $\infty$ means that the computation did
not end after $24$ hours.

\subsection{Simple example}

We first illustrate the behavior of our algorithm on the
simple planar determinantal quartic of Example \ref{ex:quartic}.
We would like to find at least one point $(x_1,x_2) \in \RR^2$ in each connected component
of the real variety defined by the equation
\[
\begin{array}{rcl}
\det\left(\begin{array}{cccc}
1+\X_1  & \X_2   & 0 & 0 \\
  \X_2  & 1-\X_1 & \X_2 & 0 \\
     0  & \X_2   & 2+\X_1 & \X_2 \\
     0  & 0      & \X_2 & 2-\X_1
\end{array}\right)  & = & \\\\
\X_1^4 + 3\X_1^2\X_2^2 +\X_2^4 -\X_1\X_2^2 - 5\X_1^2 - 7\X_2^2 + 4 & = & 0.
\end{array}
\]

With input the previous linear matrix, the algorithm checks that
the associated incidence variety $\setV$ verifies the regularity properties.
This is done by computing a Gr\"obner basis of the ideal generated
by the polynomials defining $\setV$ and by the maximal minors of the
Jacobian matrix, and verifying that this Gr\"obner basis is 1. \smallskip

Then, the algorithm recursively computes rational parametrizations of
the zero-dimen\-sional Lagrange systems encoding critical points of the
projection on the first variable, restricted to the incidence varieties
(or its sections). To obtain this parametrization, we use the functions
implemented in the Maple package {\tt fgbrs} given in input a
Gr\"obner basis of a zero-dimensional ideal, gives in output a
rational parametrization of its solution set.

\begin{figure}[!ht]
\centering
\includegraphics[width=0.7\textwidth]{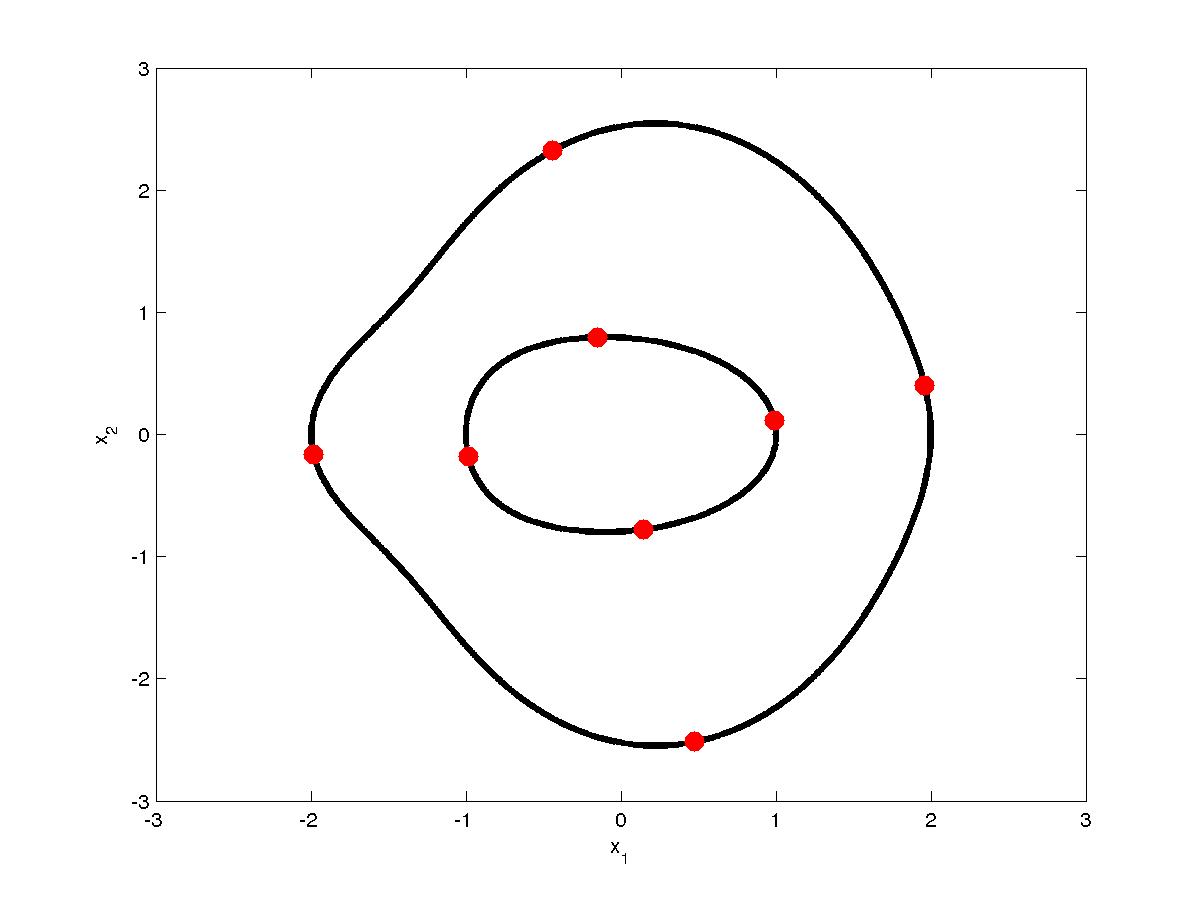}
\caption{The determinantal quartic curve of Example \ref{ex:quartic} (black)
and eight of its points (red) as returned by \textsf{RealDet}.\label{fig:quarticpoints}}
\end{figure}      

Once a rational parametrization of the desired output is given, we isolate 
the real roots which are given by isolating intervals, each of one
guaranteed to contain a point on the curve. To give an idea of the
output, we reproduce here one of these points, together with its
approximation to 10 certified digits:
\[
\begin{array}{l}
x_1 \in [\frac{122156404883928000480132795924333}{256536504662931063109335249846272},
\frac{355364086934036023530184499052519}{746288013564890365408975272280064}]
\approx  0.4761755254\\[.3em]
x_2 \in [-\frac{10810534239}{4294967296}, -\frac{345937095647}{137438953472}]
\approx -2.517023645
\end{array}
\]
The eight points are represented on the curve on Figure \ref{fig:quarticpoints}.

\subsection{Timings}

Table \ref{tab:timings} reports on timings obtained with $n$-variate
linear matrices of size $m$ with rational
coefficients chosen randomly. Thus, all matrices satisfy the genericity
Assumption $\sfG$.

\begin{table}[!ht]
  \centering
  \begin{tabular}{|c|c|c|c||||c|c|c|c|}
    \hline
    $m$ & $n$ & {\sf RealDet} & {\sc RAGlib} & $m$ & $n$ & {\sf RealDet} & {\sc RAGlib} \\
    \hline
    \hline
    $2$ & $4$    & 0.22 s &  2.25 s      &  $4$ & $3$   & 4.16 s & 2.15 s   \\
    $2$ & $10$   & 0.63 s &  25.6 s      &  $4$ & $4$   &  110 s &  835 s   \\ 
    $2$ & $20$   & 1.99 s &  $\simeq$ 1 h    &  $4$ & $8$   & 1824 s & $\infty$ \\
    $3$ & $3$    & 0.49 s &   2.8 s      &  $4$ & $16$  & 4736 s & $\infty$ \\
    $3$ & $9$    & 2.24 s &   195 s      &  $4$ & $20$  & 7420 s & $\infty$ \\
    $3$ & $20$   & 10.5 s & $\simeq$ 7 h &  $5$ & $2$   &  0.9 s & 0.23 s   \\
    $4$ & $2$    & 0.35 s &  0.35 s      &  $5$ & $3$   & 10.2 s &    59 s   \\  
    \hline
  \end{tabular}
  \vspace{0.4cm}
  \caption{Timings for \textsf{RealDet} applied to random linear matrices}
  \label{tab:timings}
\end{table}

We can observe that our implementation  {\sf RealDet} reflects the complexity gain
since, for example, we are able to solve the problem for dense
determinants of degree $m=4$ and with $n=16$ variables in less than one
hour and a half; the same problem cannot be solved within a day
by {\sc RAGlib}. 

\begin{figure}[!ht]
\centering
\begin{tikzpicture} 
  \begin{semilogyaxis} [xlabel={$n$: number of variables}, ylabel=time (s), legend style={anchor=west, at={(0.6,0.7)}}]
    \addplot coordinates
             {(1, 0.3) (2, 0.35) (3, 0.49)
              (4, 0.67) (5, 0.88) (6, 1.03)
               (7, 1.39) (8, 1.85) (9, 2.24)
                (10, 2.7) (11, 2.99) (12, 3.5)
               (13, 4) (14, 4.7) (15, 6)
               (16, 7.2) (17, 8.5) (18, 9.4)
               (19, 10) (20, 10.5)
             };
    \addplot coordinates
             {(1, 0.1) (2, 0.156) (3, 1.41)
               (4, 5.222) (5, 50.584) (6, 200.512)
               (7, 800.514) (8, 2605) (9, 11128.875)
             };
             \legend{$\mathsf{RealDet}$, {\sc RAGlib}}
  \end{semilogyaxis}
\end{tikzpicture}
\caption{Timings for $m=3$ and $n \leq 20$}
\label{fig:timings}
\vspace{1cm}
\begin{tikzpicture} 
  \begin{semilogyaxis} [xlabel={$n$: number of variables}, ylabel=time (s), legend style={anchor=west, at={(0.6,0.5)}}]
    \addplot coordinates
             {(1, 2) (2, 3) (3, 4.16)
               (4, 110) (5, 200) (6, 450)
               (7, 980) (8, 1824) (9, 2300)
               (10, 2600) (11, 3100) (12, 3600)
               (13, 3999) (14, 4250) (15, 4600)
               (16, 4736) (17, 5200) (18, 5790)
               (19, 6500) (20, 7420)
             };
             \addplot coordinates
                      {(1, 0.2) (2, 0.35) (3, 2.15)
                        (4, 435) (5, 1696) (6, 3000)
                        (7, 6024) (8, 15500)
                      };
                      \legend{$\mathsf{RealDet}$, {\sc RAGlib}}
  \end{semilogyaxis}
\end{tikzpicture}
\caption{Timings for $m=4$ and $n \leq 20$}
\label{fig:timings2}
\end{figure}

Also, when the size $m$ of the matrix is fixed, we observe that the
increase of time needed to perform the computation is well-controlled.
Figures \ref{fig:timings} and \ref{fig:timings2} illustrate this: the black (resp. red) curve
represents how the computation time of our implementation (resp. {\sc
 RAGlib}) increases with respect to the number of variables when $m$
is fixed to $3$ and $4$. Note that our implementation has the ability
to solve problems with $20$ variables which are unreachable by {\sc RAGlib}.

\subsection{Degree of the output}

In Table \ref{tab:degree}, we report some data on the degrees of the
rational parametrizations computed by ${\sf RealDet}$. Recall
that we have provided degree bounds in Section
\ref{ssec:algo:complexity}.

\begin{table}[!ht]
  \centering
  \begin{tabular}{|c|c|c|c||||c|c|c|c||||c|c|c|c|}
    \hline
    $m$ & $n$ & degree & bound & $m$ & $n$ & degree & bound & $m$ & $n$ & degree & bound \\
    \hline
    \hline
    $2$ & $2$  & 4  & 5  & $3$ & $4$   & 33 & 43 & $4$ & $3$   & 52  & 74  \\
    $2$ & $3$  & 6  & 7  & $3$ & $5$   & 39 & 49 & $4$ & $4$   & 120 & 169 \\
    $2$ & $4$  & 6  & 7  & $3$ & $6$   & 39 & 49 & $4$ & $6$   & 264 & 347 \\ 
    $2$ & $8$  & 6  & 7  & $3$ & $8$   & 39 & 49 & $4$ & $7$   & 284 & 367 \\
    $2$ & $20$ & 6  & 7  & $3$ & $15$  & 39 & 49 & $4$ & $15$  & 284 & 367 \\
    $3$ & $3$  & 21 & 28 & $3$ & $20$  & 39 & 49 & $4$ & $20$  & 284 & 367 \\
    \hline
  \end{tabular}
  \vspace{0.4cm}
  \caption{Degree of the output for the generic case}
  \label{tab:degree}
\end{table}

We conjectured that these bounds are not
sharp; these experiments support this statement.  In the column
``degree'' we report the sum of the degrees of the rational
parametrizations computed by our algorithm for generic $n$-variate
linear matrices of size $m$. We remark that if $m$ is fixed, this
value is constant when $n \geq 2m-1$. The same property holds for the
multi-linear bound for the degree of the output.

\begin{example}
Consider the matrix
$$
\A(\X) =
\left(\begin{array}{cccc}
  \X_{11} & \X_{12}  & \ldots & \X_{1m}  \\
  \X_{21} & \ddots  &        & \vdots  \\
  \vdots &         &        &         \\
  \X_{m1} &         &        & \X_{mm}
\end{array}\right).
$$
We remark that, in the context of this paper, $\A(\X)$ is a linear matrix
of size $m$, with $m^2$ variables, and it is expressed as a linear combination of $m^2$
matrices of rank 1. Allowing $\X \in \QQ^{m^2}$ to vary, the matrix
$\A(\X)$ describes all matrices of size $m$ with entries in $\QQ$.

Let $b = (b_{11} \ldots b_{mm}) \in \QQ^{m^2}$ be a vector of
rational numbers. We add the affine constraint $b'\X = 1$, i.e.
we solve the previous linear equation with respect to $\X_{11}$ and we
substitute this value to $\X_{11}$ into $\A(\X)$.

\begin{table}[!ht]
  \centering
  \begin{tabular}{|c||c|c|c|}
    \hline
    $b$ all ones & $m=2$ & $m=3$ & $m=4$  \\
    \hline
    \hline
    degree   & 5     & 35  & 244  \\
    \hline
  \end{tabular}
  \hspace{1cm}
  \begin{tabular}{|c||c|c|c|}
    \hline
    $b$ generic & $m=2$ & $m=3$ & $m=4$  \\
    \hline
    \hline
    degree   & 6   & 36  & 245  \\
    \hline
  \end{tabular}
  \vspace{0.4cm}
  \caption{Matrices with an affine constraint on the entries}
  \label{tab:classic}
\end{table}

In Table \ref{tab:classic} we report on some numerical experiments.
The two subtables contains the degree of the output
of ${\sf RealDet}$ and the computational times
respectively when $b$ is the vector of all ones, and when the coordinates of $b$
 are random values in $\QQ$. We remark that
the values of the degree are smaller than the
  corresponding values for the ``dense'' cases $(m,n) = (2,3), (3,8)$ and $(4,15)$
  that are respectively $6, 39$ and $284$, as shown in Table \ref{tab:degree}.
\end{example}

\begin{example}
Consider the symmetric matrix
$$
\A(\X) =
\left(\begin{array}{cccc}
  2\X_{11} & \X_{12}  & \ldots & \X_{1k} \\
  \X_{12}  & \ddots  &        & \vdots  \\
  \vdots  &         &        &         \\
  \X_{1k}  &         &        & 2\X_{kk}
\end{array}\right).
$$

Matrix $\A(\X)$ has size $m$  with $m(m+1)/2$ variables and
it parametrizes all symmetric matrices. It is expressed as a linear
combination of matrices of rank $1$ or $2$.

\begin{table}[!ht]
  \centering
  \begin{tabular}{|c||c|c|c|}
    \hline
    $b$ all ones & $m=2$ & $m=3$ & $m=4$  \\
    \hline
    \hline
    degree  & 2 & 16 & 122 \\
    \hline
  \end{tabular}
  \hspace{1cm}
  \begin{tabular}{|c||c|c|c|}
    \hline
    $b$ generic & $m=2$ & $m=3$ & $m=4$  \\
    \hline
    \hline
    degree   & 3 & 21 & 136  \\
    \hline
  \end{tabular}
  \vspace{0.4cm}
  \caption{Symmetric matrices with an affine constraint on the entries}
  \label{tab:symmetric}
\end{table}

We add as above a linear relation $b'\X=1$
where $b \in \QQ^{m(m+1)/2}$, and in Table \ref{tab:symmetric} we
report on experimental data. We observe the same behavior
as in the previous example.
\end{example}

\subsection{Complexity}

In Figures \ref{fig:compl} and \ref{fig:compl2}, we consider two fundamental
subclasses of the problem: when $n=m^2$ (non-symmetric case)
and when $n=m(m+1)/2$ (symmetric case).
We estimate in both cases the order of complexity 
$$
C(m,n) = n^2m^2(n+m)^5{m+n \choose n}^6
$$
of {\sf RealDet} as computed in Proposition \ref{prop:complexity:RealDet}.
We recall that standard complexity bounds for these classes of
problems are in $m^{\mathcal{O}(n)}$.

\begin{figure}[!ht]
\centering
\includegraphics{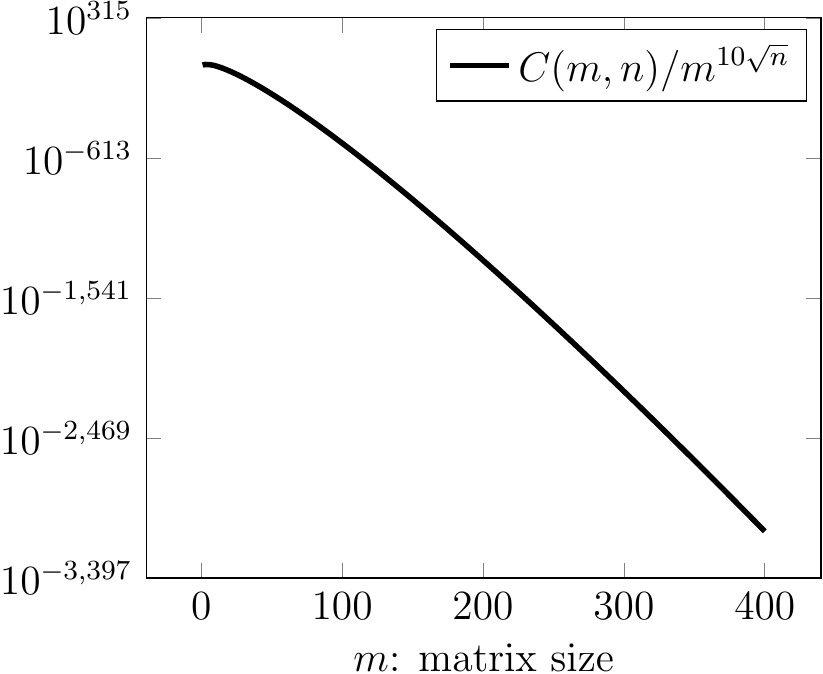}
\caption{Complexity bound for $n=m^2$.}
\label{fig:compl}
\end{figure}


On Figure \ref{fig:compl} we represent in logarithmic scale the ratio of 
$C(m,n)$ with $m^{10 \sqrt{n}}$ (where the relation $n=m^2$ is fixed)
as a function of the matrix size $m$. We remark that we obtain a bound which
is strictly contained in $m^{\mathcal{O}(\sqrt{n})}$ since this ratio tends to zero.
This numerical test shows that our complexity bound, significantly
improves the previous one. 

\begin{figure}[!ht]
\centering
\includegraphics{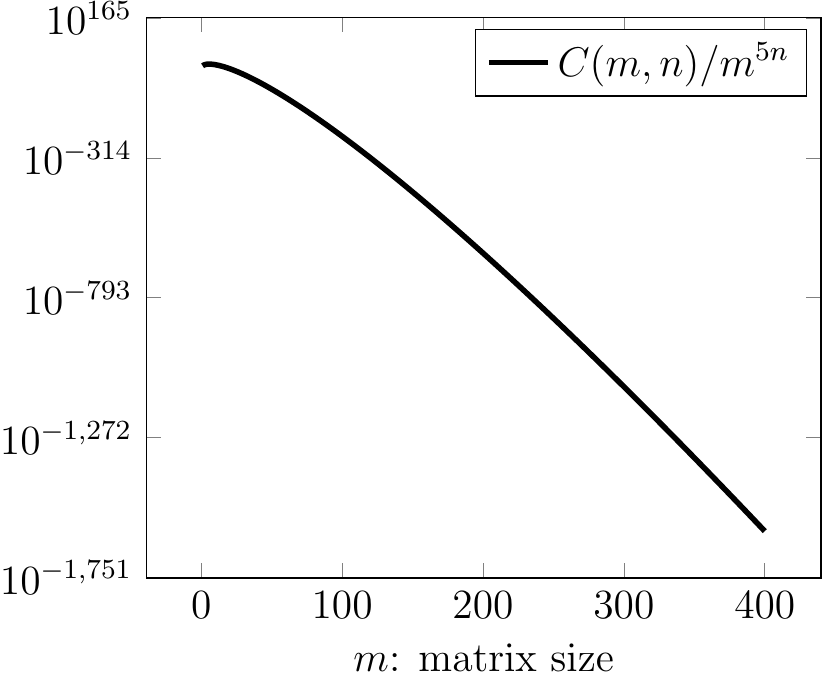}
\caption{Complexity bound for $n=\frac{m^2+m}{2}$.}
\label{fig:compl2}
\end{figure}


The same conclusion holds for the second case (Figure \ref{fig:compl2}) where $n=(m^2+m)/2$,
 which includes the fundamental family of symmetric linear matrices), where our
complexity is compared with $m^{5n}$. We also remark that similar results --
not reported here for conciseness -- have been obtained by imposing a linear
relation between $m$ and $n$, for example $n=2m$ or $n=3m$, and allowing $m$ to vary.

To summarize, the complexity of {\sf RealDet} given by Proposition \ref{prop:complexity:RealDet}
is such that:
\begin{itemize}
\item when $m$ is fixed, the complexity $n \mapsto C(m,n)$ is polynomial;
\item when $n=m^2$ or $n=(m^2+m)/2$ or $n= \alpha m$,  its asymptotic behavior when $m$ grows
is well-controlled and improves the state-of-the-art.
\end{itemize}


\end{document}